\documentclass[journal]{IEEEtran}
\usepackage{amsfonts}

\ifCLASSINFOpdf \else \fi
\usepackage{cuted}
\usepackage{cite}
\usepackage{amsthm,amsmath,amssymb,amsfonts}
\usepackage{amssymb}
\usepackage{algorithmic}
\usepackage{graphicx}
\usepackage{textcomp}
\usepackage{xcolor}
\usepackage{amsmath}
\usepackage{subfigure}
\usepackage{epstopdf}
\usepackage{xcolor}
\usepackage{makecell}\IEEEoverridecommandlockouts
\begin{document}
\newtheorem{theorem}{Theorem}
\newtheorem{corollary}{Corollary}
\newtheorem{remark}{Remark}
\title{Waveform Design for MIMO-OFDM Integrated Sensing and Communication System: \\
An Information Theoretical Approach}
\author{Zhiqing Wei,~\IEEEmembership{Member,~IEEE,}
	Jinghui Piao,~\IEEEmembership{Student Member,~IEEE},\\
	Xin Yuan,~\IEEEmembership{Member,~IEEE,}
	Huici Wu,~\IEEEmembership{Member,~IEEE},\\
	J. Andrew Zhang,~\IEEEmembership{Senior Member,~IEEE,}
	Zhiyong Feng,~\IEEEmembership{Member,~IEEE},\\
	Lin Wang,~\IEEEmembership{Student Member,~IEEE},
	and Ping Zhang~\IEEEmembership{Fellow,~IEEE}
	
	% <-this % stops a space
	\thanks{This work was supported in part by the National Key Research and Development Program of China under Grant 2020YFA0711302,
		in part by the National Natural Science Foundation of China (NSFC) under Grant 62271081, Grant 92267202 and Grant U21B2014.
		
		Zhiqing Wei, Jinghui Piao, Zhiyong Feng, Lin Wang, and
		Ping Zhang are with School of Information and Communication
		Engineering, Beijing University of Posts and Telecommunications, Beijing
		100876, China (email: weizhiqing@bupt.edu.cn; piaojinghui@bupt.edu.cn; fengzy@bupt.edu.cn;
		wlwl@bupt.edu.cn; pzhang@bupt.edu.cn).
		
		Huici Wu is with the National Engineering Lab for Mobile Network Technologies,
		Beijing University of Posts and Telecommunications, Beijing 100876,
		China (e-mail: dailywu@bupt.edu.cn).
		
		X. Yuan is with the University of Technology Sydney, Ultimo, NSW 2007,
		Australia (email: xin.yuan@ieee.org).
		
		J. A. Zhang is with the Global Big Data Technologies Centre, University
		of Technology Sydney, Sydney, NSW, Australia (email: Andrew.
		Zhang@uts.edu.au).
		
		Correspondence authors: Jinghui Piao, Zhiqing Wei.
	%\vspace{-2cm}
	}}
\maketitle

\begin{abstract}
Integrated sensing and communication (ISAC) is regarded as the enabling technology 
in the future 5th-Generation-Advanced (5G-A) and 6th-Generation (6G) 
mobile communication system. 
ISAC waveform design is critical in ISAC system.
However, the difference of the performance metrics between 
sensing and communication brings challenges for the ISAC waveform design.
This paper applies the unified performance metrics in information theory, 
namely mutual information (MI), to measure the communication and sensing performance in multicarrier ISAC system. 
In multi-input multi-output orthogonal frequency division multiplexing (MIMO-OFDM) ISAC system, we first derive the sensing and communication MI with subcarrier correlation and spatial correlation. 
Then, we propose optimal waveform designs for maximizing the sensing MI, communication MI and the weighted sum of sensing and communication MI, respectively.
The optimization results are validated by Monte Carlo simulations.
Our work provides effective closed-form expressions for waveform design, enabling the realization of MIMO-OFDM ISAC system with balanced performance in communication and sensing.
\end{abstract}


\section{Introduction}

Driven by emerging intelligent applications such as smart city, 
autonomous driving, etc., 
which require both high-quality wireless connections and high-accurate sensing capability~\cite{wang2021},
the sensing function is integrating into the future 5th-Generation-Advanced (5G-A) and 6th-Generation (6G) mobile communication networks. 
Integrated sensing and communication (ISAC) \cite{liu2022}, which realizes sensing and communication functions with the same hardware and wireless resources \cite{feng2020}, has great potential in supporting the intelligent applications in 5G-A and 6G \cite{imt2030}.

The ISAC waveform design is fundamental for the ISAC system \cite{zhang2021}, 
which has a crucial influence on the sensing and communication performance. 
However, the ISAC waveform design faces the challenge of the difference in performance metrics between sensing and communication \cite{pin2021}. 
The performance metrics of communication are mainly spectrum efficiency, 
bit error rate, delay, etc. 
While the performance metrics of sensing are mainly 
detection performance metric, including false alarm probability and detection probability, and estimation performance metric, including Minimum Mean Square Error~(MMSE) and Cramer-Rao Lower Bound (CRLB), etc.~\cite{kay1993}.
The difference of the performance metrics between sensing and communication makes it difficult to be applied to form a unified performance indicator 
to measure the overall performance of the ISAC system \cite{pin2021}, 
which brings challenges for the ISAC waveform design.

The unified performance metrics of the ISAC system have been studied in recent years, including estimation information rate \cite{b1}, equivalent Mean Square Error (MSE) \cite{b2}, capacity distortion \cite{b3}\cite{b4}, mutual information (MI) \cite{b5,b6,b7,b8}, etc., which are introduced as follows.

\begin{itemize}
\item Estimation information rate is an approximate MI between 
observation and parameter. 
Chiriyath \textit{et al.} \cite{b1} regarded the jointly received signals as multiple access channels and adopted estimation information rates to predict time delays.
With the communication information rate characterizing the performance of communication, the tradeoff between communication and sensing was studied. 

\item Equivalent MSE is an MSE metric that reveals the tradeoff between estimation MSE and communication MSE by considering the MMSE of estimating input from the output in a Gaussian channel.
In \cite{b2}, an effective communication MMSE based on rate-distortion theory was proposed, which is combined with CRLB to provide the performance metrics for the ISAC system.

\item Capacity distortion is the maximum rate at which transmission rate and state estimation distortion are achievable \cite{b3}, which measures the tradeoff between sensing and communication. 
The capacity distortion tradeoff for the memoryless channel has been studied in \cite{b4} to maximize the probability distribution of the channel input symbols.

\item MI has been applied to evaluate the performance of communication and sensing in the ISAC system recently \cite{b5}. Sensing MI is the conditional MI between the sensing channel and echo signal \cite{b6}. While communication MI is defined as the conditional MI between the transmitted and received signals. Bell \cite{b6} first derived the conditional MI between the target response and the echo signal. Yang \emph{et al.} \cite{b7} proved that the minimization of MMSE is equivalent to the maximization of MI. Hence, sensing MI is widely adopted as the performance metric of sensing \cite{b7}\cite{b8}. Both sensing MI and communication MI characterize the limit of information acquisition. Compared to other performance metrics, MI provides the similar physical significance, and mathematical expressions for sensing and communication \cite{b5}, as well as the same unit of measurement and similar optimization techniques \cite{b5}. Additionally, MI is not limited to perfect scenario assumptions (e.g., CSI imperfectness \cite{b12} and environmental clutter \cite{liutransmission}), specific algorithms, etc.
Therefore, MI is widely used as a unified performance metric for the ISAC system, allowing tractability and robustness for the performance evaluation.
% Hence, as a unified performance metric for the ISAC system allowing robustness and tractability for the performance evaluation, MI is widely applied in ISAC signal design.
\end{itemize}

As a unified ISAC performance metric, MI is adopted in ISAC waveform design, including ISAC waveform design in the time-frequency domain and space domain \cite{b9,b10,b11,b12}. For ISAC waveform design in the time-frequency domain, Zhu \emph{et al.} \cite{b9} proposed an OFDM radar-communication waveform optimization method. The objective function is the interception probability, and the constraints are sensing MI and data information rate. Ahmed \emph{et al.} \cite{b10} maximized the overall and worst-case communication MI under the constraint of transmit power or sensing MI to realize power and subcarriers allocation of OFDM signals. In \cite{b11}, the sensing MI and Data Information Rate (DIR) are jointly maximized for OFDM ISAC waveform under the constraint of the total transmit power. For the ISAC waveform design in the space domain, most studies focused on MIMO technologies. Yuan \emph{et al.} \cite{b12} proposed a joint space-time optimization scheme to maximize the weighted sum of sensing and communication MIs with spatial correlation MIMO considering channel estimation errors. However, existing studies did not provide the analysis of sensing and communication MI in the space-time-frequency domains, which is the theme of this paper.

This paper derives the closed-form sensing and communication MI expressions for the MIMO-OFDM ISAC downlink system in the perceptive mobile networks \cite{zhang2021}, 
where the base station (BS) utilizes the echo of the downlink communication signal for radar sensing. For communication, the instantaneous communication MI is derived with the assumption that both the transmitter and communication receiver have perfect knowledge of the communication channels. For sensing, the instantaneous sensing MI is derived with the assumption that both the transmitter and sensing receiver possess perfect knowledge of the transmitted signals. The ISAC waveform design is provided based on the derived sensing and communication MIs. It should be noted that although the waveform is optimized in this paper, the actual transmitted information is not changed, since the designed waveforms with the desired correlation can be realized by the techniques, such as precoding and signal modulations. The main contributions of this paper are summarized as follows.
\begin{itemize}
	\item Based on the MIMO-OFDM multi-target sensing model, the correlation between the subcarriers in the sensing channel is modeled and analyzed, and the closed-form sensing MI of the MIMO-OFDM system with subcarrier correlation is derived.
	
	\item Based on the frequency-selective fading MIMO-OFDM communication model, the analytical and closed-form communication MIs for the MIMO-OFDM system under perfect channel state information (CSI) are derived. 
 
    \item Based on the derived sensing and communication MIs, and the correlation matrices between the subcarriers in the sensing channel, the expression for the optimized waveform is derived based on the flexible weighted sum of the sensing and communication MIs. The tradeoff between sensing and communication MIs is revealed.
    
\end{itemize}

The rest of this paper is organized as follows. In Sec. \uppercase\expandafter{\romannumeral2}, the communication and sensing models are described. In Sec. \uppercase\expandafter{\romannumeral3}, we derive the closed-form expressions of communication MI and sensing MI. Sec. \uppercase\expandafter{\romannumeral4} discusses three power allocation schemes using the derived MI. The simulation results are provided in Sec. \uppercase\expandafter{\romannumeral5}. The key notations used in this paper are listed in Table \ref{sys_para}.

\begin{table}[!t]
	\renewcommand\arraystretch{1.2}
	\caption{\label{sys_para}Key Notations and Definitions}
	\begin{center}
		\begin{tabular}{l l}
			\hline
			\hline
			
			{Notations} & {Definitions} \\
			
			\hline
			
			${N}_{c}$ & Number of subcarriers  \\
			${N}_{t}$ & Number of transmit antenna elements \\
			${{N}_{r}}$ & Number of receive antenna elements  \\
			${{N}_{x}}$ & Number of OFDM symbols\\
			${L}_{c}$ & Number of paths\\
			${L}_{r}$ & Number of targets\\
			$H_{\mu\nu}(p)$ &Communication channel coefficient for $p$-th subcarrier \\&from $\mu$-th transmit antenna to $\nu$-th receive antenna\\
			$G_{\mu\nu}(p)$ &Sensing channel coefficient for $p$-th subcarrier from $\mu$-th \\&transmit antenna to $\nu$-th receive antenna\\
			$E$ & Average transmit power\\
			${\omega }_{r}$ & Sensing weighting coefficient\\
			${\omega }_{c}$ & Communication weighting coefficient\\
			${F}_{r}$ & Maximum sensing MI\\
			${F}_{c}$ & Maximum communication MI\\
			${F}_{\omega}$ & Maximum weighted MI\\
			$\sigma_n^2$ & Noise variance \\
			
			\hline
			\hline
		\end{tabular}
	\end{center}
\end{table}

\textit{Notation:} $\mathbf X$ denotes matrix and ${\mathbf x}_{i}$ denotes $i$-th column of $\mathbf X$. $(\cdot)^{T}$, $(\cdot)^{H}$ and $(\cdot)^{-1}$ denotes transposition, conjugate transportation and inverse, respectively. Besides, $\mathbf I$ is the identity matrix. $E$[$\cdot$] denotes the expectation operator. det$(\cdot)$ and tr$(\cdot)$ are the determinant and trace of matrix, respectively. diag$ {\left\lbrace  x_1,x_2,\cdots,x_n\right\rbrace }$ is diagonal matrix with diagonal elements $x_1,x_2,\cdots,x_n$ and diag$ {\left\lbrace  {\mathbf X}_1,{\mathbf X}_2,\cdots,{\mathbf X}_n\right\rbrace }$ is block diagonal matrix with diagonal matrices ${\mathbf X}_1,{\mathbf X}_2,\cdots,{\mathbf X}_n$.

\section{System Model}

A MIMO-OFDM ISAC system comprising multiple radar targets is considered.
Node A (i.e., transmitter) communicates with node B (i.e., receiver) and senses the environment simultaneously via  ${{N}_{c}}$ subcarriers, 
as shown in Fig. 1. 
Specifically, node A transmits data to node B for communication, 
while receiving the echo signals for sensing the environment. 
Each of node A and B has two spatially separated antenna arrays, 
i.e., ${{N}_{t}}$ transmit antenna elements and ${{N}_{r}}$ receive antenna elements.  
We assume that both nodes A and B have perfect communication CSI, while only the correlation matrix of the sensing channel is known at the transmitter \cite{ouyang2022}, which is practical to obtain through field measurements, or by some feedback mechanism \cite{b7}.
If the ISAC system has imperfect CSI, the communication channel will be affected by the channel estimation error ${\sigma_e}^2$ and the derivation of communication and sensing MIs will be affected by the approach of allocating training and data payload signals, channel estimation error, etc., which will be investigated in the future work.

The data on the $\mu $-th transmit antenna element 
during the $n$-th symbol interval time on the $p$-th subcarrier 
is $x_{n}^{\mu }(p),p=0,1,\cdots ,{{N}_{c}}-1$. 
$\{x_{n}^{\mu }(p),p=0,1,\cdots ,{{N}_{c}}-1\}$ are transmitted in parallel on ${{N}_{c}}$ subcarriers through the fading channel with undergoing Inverse Fast Fourier Transform (IFFT) and adding a cyclic prefix.
\begin{figure}[t]
	\centering
	\includegraphics[width=0.5\textwidth]{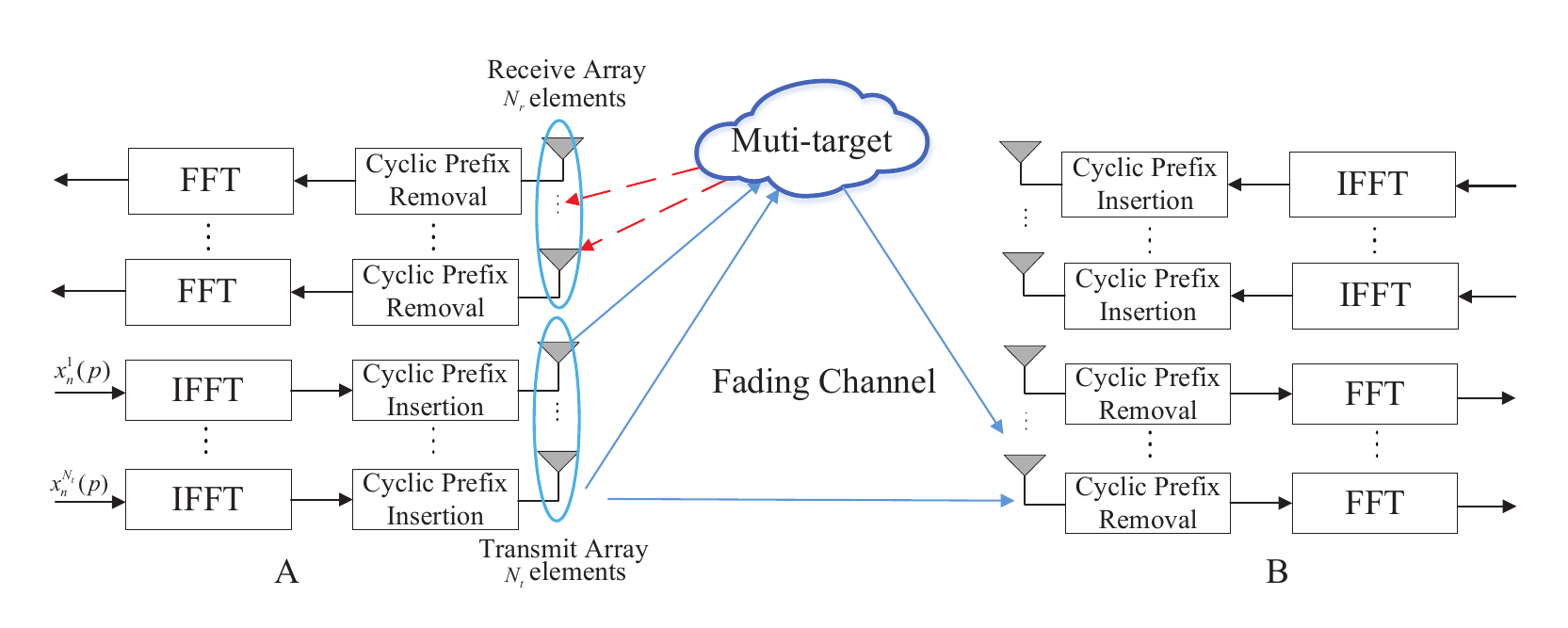}
	\caption{MIMO-OFDM ISAC system.}
	\label{fig}
\end{figure}

\subsection{Communication Models}

It is assumed that the MIMO-OFDM channel is frequency-selective and time-flat with a channel order ${{L}_{c}}$, where ${{L}_{c}}$ is the number of paths. 
The impulse response of the wireless channel is \cite{ts2022}
\begin{equation}
	h(t,\tau )=\sum\limits_{l=0}^{{{L}_{c}}}{{{h}_{l}}(t)\delta (\tau -{{\tau }_{l}})}\label{eq1},
\end{equation}
where ${{h}_{l}}(t)$ is the impulse response for the $l$-th path, which is independent of other  $(L_c - 1)$ paths. The channel gains of the $l$-th MIMO path from ${{N}_{t}}$ transmit antenna elements to ${{N}_{r}}$ receive antenna elements are expressed as the following discrete baseband equivalent impulse response matrix.
\begin{equation}
	{{\mathbf{H}}_{l}}=\left( \begin{matrix}
		{{h}_{11}}(l) & \ldots  & {{h}_{1{{N}_{r}}}}(l)  \\
		\vdots  & \ddots  & \vdots   \\
		{{h}_{{{N}_{t}}1}}(l) & \cdots  & {{h}_{{{N}_{t}}{{N}_{r}}}}(l)  \\
	\end{matrix} \right)\in {{\mathbb{C}}^{{{N}_{t}}\times {{N}_{r}}}}\label{eq2},
\end{equation}
where ${{h}_{\mu \nu }}(l)$ is the channel time domain coefficient between the $\mu $-th transmit antenna element and the $\nu $-th receive antenna element on the $l$-th path, including the influence of the transmit and receive filter and the relative delay between the antennas. 

It is supposed that the receiver antenna elements are spatially uncorrelated due to the presence of multiple local scatters around the receiver or widely separated placement \cite{Jafar}.
In contrast, the spatial correlation of the transmit antenna elements still 
needs to be considered. 
As the channel matrix of different paths is independent, 
the $l$-th path of the random MIMO channel response can be decomposed into \cite{b15}
\begin{equation}
	{{\mathbf{H}}_{l}}=\mathbf{R}_{l}^{\frac{1}{2}}{{\mathbf{H}}_{\omega ,l}}\label{eq3},
\end{equation}
where ${{\mathbf{R}}_{l}}=E\{{{\mathbf{h}}_{\nu }}(l)\mathbf{h}_{\nu }^{H}(l)\}$ ($l=0,\cdots ,{{L}_{c}};\text{ }\nu =1,\cdots ,{{N}_{r}}$) represents the correlation coefficient matrix of the antenna on the $l$-th path and ${{\mathbf{h}}_{\nu }}(l)$ stands for the column vector of ${{\mathbf{H}}_{l}}$. The elements in ${{\mathbf{H}}_{\omega ,l}}\in {{\mathbb{C}}^{{{N}_{t}}\times {{N}_{r}}}}$ follow independent and identically distributed (i.i.d.) zero-mean circularly symmetric complex Gaussian (CSCG) distribution with variance one, i.e., $\mathcal{C}\mathcal{N}(0,1)$ \cite{b15}.

The communication receivers are assumed to have perfect carrier synchronization, timing, and sampling of the symbol rate. Once the cyclic prefix has been removed and the Fast Fourier Transform (FFT) has been completed, the received data on the $\nu $-th receive antenna element are given by
\begin{equation}
	y_{\text{com},n}^{\nu }(p)=\sum\limits_{\mu =1}^{{{N}_{t}}}{{{H}_{\mu \nu }}}(p)x_{n}^{\mu }(p)+w_{\text{com},n}^{\nu }(p)\label{eq4},
\end{equation}
where ${{H}_{\mu \nu }}(p)=\sum\nolimits_{l=0}^{L_c}{{{h}_{\mu \nu }}(l){{e}^{-j(2\pi /{{N}_{c}})lp}}}$, and $w_{\text{com},n}^{\nu }(p)$ follows the i.i.d. zero-mean CSCG distribution with variance $\sigma _{n}^{2}$, i.e.,  $\mathcal{C}\mathcal{N}(0,\sigma _{n}^{2})$.

For the convenience of subsequent derivation, the subscripts of the antenna are omitted. With OFDM symbols, the received signal at the $p$-th subcarrier, denoted by ${{\mathbf{Y}}_{\text{com}}}(p)$, is 
\begin{equation}
	{{\mathbf{Y}}_{\text{com}}}(p)=\mathbf{X}(p)\mathbf{H}(p)+{{\mathbf{W}}_{\text{com}}}(p)\label{eq5},
\end{equation}
where $\mathbf{H}(p)\in{{\mathbb{C}}^{{{N}_{t}}\times {{N}_{r}}}}$ is the frequency domain channel coefficient matrix of the channel at $p$-th subcarrier, $\mathbf{X}(p)$ is the transmit data at $p$-th subcarrier, and ${{\mathbf{W}}_{\text{com}}}(p)$ is the additive white Gaussian noise (AWGN) at the $p$-th subcarrier, as given by
\begin{equation}
		\mathbf{H}(p)=\sum\nolimits_{l=0}^{L_c}{{{\mathbf{H}}_{l}}{{e}^{-j(2\pi /{{N}_{c}})lp}}}\in {{\mathbb{C}}^{{{N}_{t}}\times {{N}_{r}}}},
\end{equation}
\begin{equation}
		\mathbf{X}(p)=\left[ \begin{matrix}
			x_{0}^{1}(p) & \ldots  & x_{0}^{{{N}_{t}}}(p)  \\
			\vdots  & {} & \vdots   \\
			x_{{{N}_{x}}-1}^{1}(p) & \ldots  & x_{{{N}_{x}}-1}^{{{N}_{t}}}(p)  \\
		\end{matrix} \right]\in {{\mathbb{C}}^{{{N}_{x}}\times {{N}_{t}}}}, 
\end{equation}
\begin{equation}
{{\mathbf{Y}}_{\text{com}}}(p)=\left[ \begin{matrix}
			y_{\text{com},0}^{1}(p) & \ldots  & y_{\text{com},0}^{{{N}_{r}}}(p)  \\
			\vdots  & {} & \vdots   \\
			y_{\text{com},{{N}_{x}}-1}^{1}(p) & \ldots  & y_{\text{com},{{N}_{x}}-1}^{{{N}_{r}}}(p)  \\
		\end{matrix} \right]\in {{\mathbb{C}}^{{{N}_{x}}\times {{N}_{r}}}},
\end{equation}
\begin{equation}
{{\mathbf{W}}_{\text{com}}}(p)\!=\!\left[ \begin{matrix}
			w_{\text{com},0}^{1}(p) & \ldots  & w_{\text{com},0}^{{{N}_{r}}}(p)  \\
			\vdots  & {} & \vdots   \\
			w_{\text{com},{{N}_{x}}-1}^{1}(p) & \ldots  & w_{\text{com},{{N}_{x}}-1}^{{{N}_{r}}}(p)  \\
		\end{matrix} \right]\!\in\! {{\mathbb{C}}^{{{N}_{x}}\times {{N}_{r}}}}.
\end{equation}

Thus, the received signal in matrix form is 
\begin{equation}
	{{\mathbf{Y}}_{\text{com}}}=\mathbf{XH}+{{\mathbf{W}}_{\text{com}}},\label{eq6}
\end{equation}
where the communication channel matrix is
\begin{equation}\label{eqh}
\mathbf H={{[{{\mathbf H}^{T}}(0),\cdots ,{{\mathbf H}^{T}}({{N}_{c}}-1)]}^{T}}\in {{\mathbb{C}}^{{{N}_{c}}{{N}_{t}}\times {{N}_{r}}}}, 
\end{equation}
the transmit data matrix is
\begin{equation}
\mathbf{X}=\text{diag}\{\mathbf{X}(0),\cdots ,\mathbf{X}({{N}_{c}}-1)\}\in {{\mathbb{C}}^{{{N}_{c}}{{N}_{x}}\times {{N}_{c}}{{N}_{t}}}}, 
\end{equation}
the received communication signal matrix is 
\begin{equation}
{{\mathbf{Y}}_{\text{com}}}={{[\mathbf{Y}_{\text{com}}^{T}(0),\cdots ,\mathbf{Y}_{\text{com}}^{T}({{N}_{c}}-1)]}^{T}}\in {{\mathbb{C}}^{{{N}_{c}}{{N}_{x}}\times {{N}_{r}}}}, 
\end{equation}
and the noise matrix is
\begin{equation}
{{\mathbf{W}}_{\text{com}}}={{[\mathbf{W}_{\text{com}}^{T}(0),\cdots ,\mathbf{W}_{\text{com}}^{T}({{N}_{c}}-1)]}^{T}}\in {{\mathbb{C}}^{{{N}_{c}}{{N}_{x}}\times {{N}_{r}}}}\!.
\end{equation}
% According to (\ref{eq1}), the fading is independent of different paths and follows CSCG distribution, i.e., $E[{{h}_{\mu \nu }}(l){{h}_{\mu \nu }}({l}')]=0~(l \ne l')$. When expanding (\ref{eq8}), we have
% \begin{equation}
% E\left[\!\mathbf{H}({{p}_{1}}){{\mathbf{H}}^{H}}({{p}_{2}})\! \right]=\!E\!\left[ \sum\nolimits_{l=0}^{L}{{{\mathbf{H}}_{l}}\mathbf{H}_{l}^{H}{{e}^{-j(2\pi /{{N}_{c}})l({{p}_{1}}-{{p}_{2}})}}}\!\right].\label{eq8n}
% \end{equation}
% Equation (7) shows that it is not necessary to consider the communication channel correlation between different subcarriers because there is no subcarrier correlation in (7). We calculate the correlation coefficient matrix of different subcarriers and get 
% \begin{equation}
% E\left[\!\mathbf{H}({{p}_{1}}){{\mathbf{H}}^{H}}({{p}_{2}})\! \right]=\!E\!\left[ \sum\nolimits_{l=0}^{L}{{{\mathbf{H}}_{l}}\mathbf{H}_{l}^{H}{{e}^{-j(2\pi /{{N}_{c}})l({{p}_{1}}-{{p}_{2}})}}}\!\right].
% \end{equation}

% Equation (8) shows that the correlation between the subcarriers of communication channel exists. However, with perfect {\color{blue}CSI}, it is possible to separate the data of each subcarrier at the receiver because of the orthogonality. Hence, when each subcarrier transmits data independently, the received data will also be independent of subcarriers, so that the correlation between the communication subcarrier channels has nothing to do with the mutual information between the received data and the transmitted data.

\subsection{Sensing Models}
Since radar detects the target through the line-of-sight (LoS) path and the non-LoS (NLoS) path will cause false alarms. Hence, the LoS path is considered in the sensing channel model of this paper. Assuming that there are $({{L}_{r}}+1)$ point targets, the location of the $l$-th target is revealed by the time delay ${{\tau }_{l}}$ and the velocity information of the $l$-th target is revealed by the Doppler frequency shift ${{v}_{l}}$. The sensing channel impulse response is defined as
\begin{equation}
	g(t,\tau )=\sum\limits_{l=0}^{{{L}_{r}}}{{{g}_{l}}(t)\delta (\tau -{{\tau }_{l}})}{{e}^{j2\pi {{v}_{l}}t}},\label{eq9}
\end{equation}
where ${{g}_{l}}(t)$ is the impulse response for the $l$-th path, which is independent of other $(L_r - 1)$ paths. The $l$-th sensing channel is the discrete baseband equivalent impulse response matrix, as given by
\begin{equation}
	{{\mathbf{G}}_{l}}=\left( \begin{matrix}
		{{g}_{11}}(l) & \ldots  & {{g}_{1{{N}_{r}}}}(l)  \\
		\vdots  & \ddots  & \vdots   \\
		{{g}_{{{N}_{t}}1}}(l) & \cdots  & {{g}_{{{N}_{t}}{{N}_{r}}}}(l)  \\
	\end{matrix} \right)\in {{\mathbb{C}}^{{{N}_{t}}\times {{N}_{r}}}},\label{eq10}
\end{equation}
where ${{g}_{\mu \nu }}(l)$ is the sensing channel time domain coefficient between the $\mu $-th transmit and the $\nu $-th receive antennas on the $l$-th path, including path loss, Doppler frequency shift, direction of departure (DoD), direction of arrival (DoA) and radar cross section (RCS).

 \begin{figure}[t]
	\centering
	\includegraphics[width=0.5\textwidth]{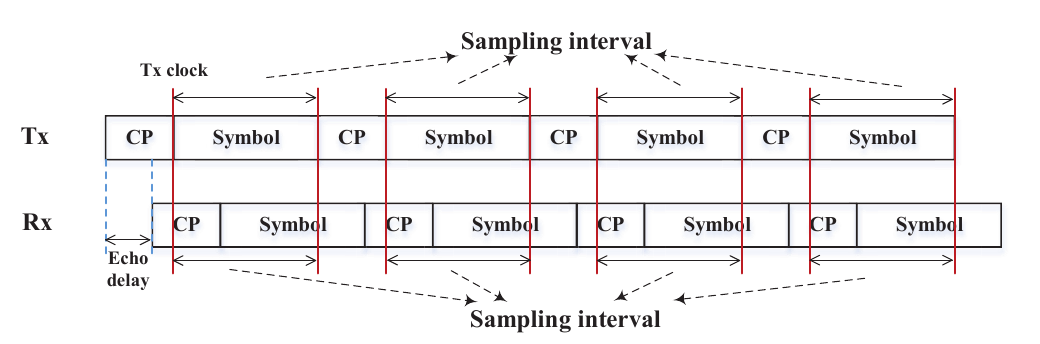}
	\caption{The removal of CP}
	\label{huifu}
\end{figure}

The radar receiver receives a frequency-dependent signal reflected by the target. The clocks of the transmitter (node A) and the radar receiver (node A) are the same and are regarded as synchronized. Besides, it is generally assumed that the echo delay does not exceed the CP length. Hence, it is reasonable to sample the received signal according to the clock of the transmitter with the time delay unknown, which is shown in Fig. 2. After removing the cyclic prefix and completing FFT processing at the radar receiver, the received echo data on the $\nu $-th receive antenna are obtained as
\begin{equation}
	y_{\text{rad},n}^{\nu }(p)=\sum\limits_{\mu =1}^{{{N}_{t}}}{G_{\mu \nu }^{n}}(p)x_{n}^{\mu }(p)+w_{\text{rad},n}^{\nu }(p),\label{eq11}
\end{equation}
where $G_{\mu \nu }^{n}(p)=\sum\nolimits_{l=0}^{L_r}{g_{\mu \nu }^{n}(l){{e}^{-j(2\pi /{{N}_{c}})lp}}}$ $(n=1,\cdots ,{{N}_{x}})$, and $w_{\text{rad},n}^{\nu }(p)$ follows the i.i.d. zero-mean CSCG distribution with variance $\sigma _{n}^{2}$, i.e.,  $\mathcal{C}\mathcal{N}(0,\sigma _{n}^{2})$.

As the sensing duration is extremely short, we assume that the sensing channel is time-invariant during the sensing period, indicating that the radar and the target remain relatively static, so that the expression of the echo signal is
\begin{equation}
	y_{\text{rad},n}^{\nu }(p)=\sum\nolimits_{\mu =1}^{{{N}_{t}}}{{{G}_{\mu \nu }}(p)x_{n}^{\mu }(p)+w_{\text{rad},n}^{\nu }(p)}.
\end{equation}

Although the time-invariant sensing channel is assumed during the sensing period, the results can also be applied to the case of the time-variant sensing channel in (\ref{eq11}).
The echo signal at the $p$-th subcarrier is 
\begin{equation}
	{{\mathbf{Y}}_{\text{rad}}}(p)=\mathbf{X}(p)\mathbf{G}(p)+{{\mathbf{W}}_{\text{rad}}}(p),\label{eq13}
\end{equation}
where $\mathbf{G}(p)=\sum\nolimits_{l=0}^{L_r}{{{\mathbf{G}}_{l}}{{e}^{-j(2\pi /{{N}_{c}})lp}}}\in {\mathbb{C}}^{{{N}_{t}}\times {{N}_{r}}}$ is the sensing channel coefficient matrix at $p$-th subcarrier. The received echo signal including ${{N}_{c}}$ subcarriers is
\begin{equation}
	{{\mathbf{Y}}_{\text{rad}}}=\mathbf{XG}+{{\mathbf{W}}_{\text{rad}}},\label{eq14}
\end{equation}
where the sensing channel matrix is
\begin{equation}
\mathbf{G}={{[{{\mathbf{G}}^{T}}(0),\cdots ,{{\mathbf{G}}^{T}}({{N}_{c}}-1)]}^{T}}\in {{\mathbb{C}}^{{{N}_{c}}{{N}_{t}}\times {{N}_{r}}}}, 
\end{equation}
the received echo signal matrix is 
\begin{equation}
{{\mathbf{Y}}_{\text{rad}}}={{[\mathbf{Y}_{\text{rad}}^{T}(0),\cdots ,\mathbf{Y}_{\text{rad}}^{T}({{N}_{c}}-1)]}^{T}}\in {{\mathbb{C}}^{{{N}_{c}}{{N}_{x}}\times {{N}_{r}}}}, 
\end{equation}
and the noise matrix is
\begin{equation}
{{\mathbf{W}}_{\text{rad}}}={{[\mathbf{W}_{\text{rad}}^{T}(0),\cdots ,\mathbf{W}_{\text{rad}}^{T}({{N}_{c}}-1)]}^{T}}\in {{\mathbb{C}}^{{{N}_{c}}{{N}_{x}}\times {{N}_{r}}}}.
\end{equation}
The channel correlation matrix of the sensing channel $\mathbf{G}$ is
\begin{equation}\label{eq15}
\begin{small}
\begin{aligned}
& E\left[ \mathbf{G}{{\mathbf{G}}^{H}} \right]=\!\!E\left\{ \left[\! \begin{matrix}
\mathbf{G}(0)  \\
\vdots   \\
\mathbf{G}({{N}_{c}}-1)  \\
\end{matrix} \right]\!\!\!\left[ {{\mathbf{G}}^{H}}(0),\cdots ,{{\mathbf{G}}^{H}}({{N}_{c}}-1) \right] \right\} \\ 
& =\!\!\left( \begin{matrix}
E[\mathbf{G}(0){{\mathbf{G}}^{H}}(0)] &\!\!\! \ldots  \!\!\!& E[\mathbf{G}(0){{\mathbf{G}}^{H}}({{N}_{c}}-1)]  \\
\vdots\!\!\!  & \ddots\!\!\!  & \vdots\!\!\!   \\
E[\mathbf{G}({{N}_{c}}-1){{\mathbf{G}}^{H}}(0)] \!\!\!&\!\!\! \ldots\!\!\!  & E[\mathbf{G}({{N}_{c}}-1){{\mathbf{G}}^{H}}({{N}_{c}}-1)]  \\
\end{matrix} \right)\!\!. \\ 
\end{aligned}
\end{small}
\end{equation}
Note that the sensing channel correlation matrix is dependent on the subcarrier correlation due to the frequency response of the frequency-selective fading channel \cite{cao}. Therefore, the correlation matrix between the different subcarriers is derived as 
\begin{equation}
\begin{split}
&E\left[ \mathbf{G}({{p}_{1}}){{\mathbf{G}}^{H}}({{p}_{2}}) \right]   \\
&=E\left[ \sum\nolimits_{l=0}^{L_r}{{{\mathbf{G}}_{l}}{{e}^{-j(2\pi /{{N}_{c}})l{{p}_{1}}}}}\sum\nolimits_{l=0}^{L_r}{\mathbf{G}_{l}^{H}{{e}^{j(2\pi /{{N}_{c}})l{{p}_{2}}}}} \right].\label{eq16}
\end{split}
\end{equation}	

According to (\ref{eq9}), the fading is independent of different paths and follows CSCG distribution, i.e., $E[{{g}_{\mu \nu }}(l){{g}_{\mu \nu }}({l}')]=0~(l \ne l')$. When expanding (\ref{eq16}), we have \cite{cao}\cite{zhu}
\begin{equation}
E\left[ \mathbf{G}({{p}_{1}}){{\mathbf{G}}^{H}}({{p}_{2}}) \right]=\!E\left[\! \sum\nolimits_{l=0}^{L_r}{{{\mathbf{G}}_{l}}\mathbf{G}_{l}^{H}{{e}^{-j(2\pi /{{N}_{c}})l({{p}_{1}}-{{p}_{2}})}}}\! \right],\label{eq17}
\end{equation}
which shows that the influence of subcarrier correlation using the sensing model in (\ref{eq14}) cannot be ignored. The subcarrier correlation decreases with the increase of subcarrier spacing. 

%For analytical tractability, different targets are assumed to have equivalent spatial correlation matrices, i.e., $E[{{\mathbf{G}}_{0}}\mathbf{G}_{0}^{H}]=\cdots=E[{{\mathbf{G}}_{L_r}}\mathbf{G}_{L_r}^{H}]=\mathbf R_\mathbf{G}$, (\ref{eq15}) can be transformed into 
%\begin{equation}\label{eq17n}
%E[\mathbf{G}\mathbf{G}^H]=(\mathbf{I}_{N_c}\otimes\mathbf R_\mathbf{G})\mathbf{\Omega}_g{\mathbf{\Omega}^H_g},
%\end{equation}
%where $\mathbf{\Omega}_g=[\omega{\color{red}_r} (0),\cdots,\omega{\color{red}_r} (N_c-1)]^T\in\mathbb{C}^{{N_c}{N_t}\times{(L_r+1)}}$, $\omega_r(p)=\left[ 1,{{e}^{-j2\pi p/{{N}_{c}}}},\ldots ,{{e}^{-j2\pi pL_r/{{N}_{c}}}} \right]^T$. (\ref{eq17n}) illustrates the dependency between different transmit antennas and subcarriers, that is, $\mathbf{R}_\mathbf{G}$ and $\mathbf{\Omega}_g{\mathbf{\Omega}^H_g}$ respectively.

\section{Mutual Information}
Based on the system model in Sec. \uppercase\expandafter{\romannumeral2}, the sensing and communication MI for the MIMO-OFDM ISAC system are derived in this section. Firstly, we derive the sensing MI considering the correlation of subcarriers. Then, we derive the closed-form communication MI. The expressions in this section provide a basis for the ISAC waveform design.
\subsection{Sensing MI}\label{AA}
Sensing MI is defined as the MI between the radar echo signal and the sensing channel given the prior knowledge of the transmitted signal, from which the mathematical expression of sensing MI can be written as
\begin{equation}
	I(\mathbf{G};{{\mathbf{Y}}_{\text{rad}}}|\mathbf{X})=h({{\mathbf{Y}}_{\text{rad}}}|\mathbf{X})-h({{\mathbf{W}}_{\text{rad}}}).\label{eq18}
\end{equation}

\begin{theorem}
    According to the MIMO-OFDM signal model in (\ref{eq14}), the sensing MI of the MIMO-OFDM ISAC system is expressed as
\begin{equation}
	\begin{split}
		&I(\mathbf{G};{{\mathbf{Y}}_{\text{rad}}}|\mathbf{X})\\
		&={{N}_{r}}{{\log }_{2}}\left[ {{(\sigma _{n}^{2})}^{-{{N}_{x}}{{N}_{c}}}}\det (\mathbf{X}{{\Sigma }_{\mathbf{G}}}{{\mathbf{X}}^{H}}+\sigma _{n}^{2}{{\mathbf{I}}_{{{N}_{x}}{{N}_{c}}}}) \right].\label{eq19}
	\end{split}
\end{equation}	
\end{theorem}
\begin{proof}
    With the knowledge of transmitted symbols, the probability density function (PDF) of received data ${{\mathbf{Y}}_{\text{rad}}}$ is
\begin{equation}
	\begin{split}
		p({{\mathbf{Y}}_{\text{rad}}}|\mathbf{X})&=\prod\limits_{j=1}^{{{N}_{r}}}{p(\mathbf{y}_{\text{rad}}^{j}|\mathbf{X})}\\
		&=\prod\limits_{j=1}^{{{N}_{r}}}{\frac{1}{{{\pi }^{{{N}_{x}}{{N}_{c}}}}\left| \mathbf{\Lambda}  \right|}\exp \left( -{{(\mathbf{y}_{\text{rad}}^{j})}^{H}}{{\mathbf{\Lambda} }^{-1}}\mathbf{y}_{\text{rad}}^{j} \right)}.
	\end{split}\label{eq20}
\end{equation}

In (\ref{eq20}), the data received on each antenna are independent of each other, so that the PDF of ${{\mathbf{Y}}_{\text{rad}}}$ is the product of the PDF of $\mathbf{y}_{\text{rad}}^{j}$ on each receive antenna. Note that $E\{\mathbf{G}{{\mathbf{G}}^{H}}\}/{{N}_{r}}=E\{{{\mathbf{G}}^{j}}{{({{\mathbf{G}}^{j}})}^{H}}\}={{\Sigma }_{\mathbf{G}}}$. Since $\mathbf{y}_{\text{rad}}^{j}$ is a complex Gaussian vector, its covariance matrix is $\mathbf{\Lambda }\text{=}\mathbf{X}{{\Sigma }_{\mathbf{G}}}{{\mathbf{X}}^{H}}$
$+\sigma _{n}^{2}{{\mathbf{I}}_{{{N}_{x}}{{N}_{c}}}}$.

Since we have
\begin{equation}
\begin{aligned}
 &{{(\mathbf{y}_{\text{rad}}^{j})}^{H}}{{(\mathbf{X}{{\Sigma }_{\mathbf{G}}}{{\mathbf{X}}^{H}}+\sigma _{n}^{2}{{\mathbf{I}}_{{{N}_{x}}{{N}_{c}}}})}^{-1}}\mathbf{y}_{\text{rad}}^{j} \\
 &=\text{tr}\left( {{(\mathbf{X}{{\Sigma }_{\mathbf{G}}}{{\mathbf{X}}^{H}}+\sigma _{n}^{2}{{\mathbf{I}}_{{{N}_{x}}{{N}_{c}}}})}^{-1}}
\mathbf{y}_{\text{rad}}^{j}{{(\mathbf{y}_{\text{rad}}^{j})}^{H}}\right), 
\end{aligned}
\end{equation}
(\ref{eq20}) can be further transformed to
\begin{equation}
	\begin{aligned}
		& p({{\mathbf{Y}}_{\text{rad}}}|\mathbf{X})\\	
		& =\frac{1}{{{\pi }^{{{N}_{x}}{{N}_{c}}{{N}_{r}}}}{{\det }^{{{N}_{r}}}}(\mathbf{X}{{\Sigma }_{\mathbf{G}}}{{\mathbf{X}}^{H}}+\sigma _{n}^{2}{{\mathbf{I}}_{{{N}_{x}}{{N}_{c}}}})}\\
		 &\times\exp \left( -\text{tr}\left( \sum\limits_{j=1}^{{{N}_{r}}}{{{(\mathbf{X}{{\Sigma }_{\mathbf{G}}}{{\mathbf{X}}^{H}}+\sigma _{n}^{2}{{\mathbf{I}}_{{{N}_{x}}{{N}_{c}}}})}^{-1}}\mathbf{y}_{\text{rad}}^{j}{{(\mathbf{y}_{\text{rad}}^{j})}^{H}}} \right)\!\!\right)\!,\label{eq21}
	\end{aligned}
\end{equation}

The correlation matrix of $\mathbf{y}_{\text{rad}}^{j}$ is
\begin{equation}
	E\{\mathbf{y}_{\text{rad}}^{j}{{(\mathbf{y}_{\text{rad}}^{j})}^{H}}\}=\mathbf{X}E\{{{\mathbf{G}}^{j}}{{({{\mathbf{G}}^{j}})}^{H}}\}{{\mathbf{X}}^{H}}+\sigma _{n}^{2}{{\mathbf{I}}_{{{N}_{x}}{{N}_{c}}}}.\label{eq22}
\end{equation}

By substituting (\ref{eq22}) into (\ref{eq21}), $h({{\mathbf{Y}}_{\text{rad}}}|\mathbf{X})$ is obtained as
\begin{equation}
	\begin{split}
		h({{\mathbf{Y}}_{\text{rad}}}|\mathbf{X})=&{{N}_{x}}{{N}_{r}}{{N}_{c}}{{\log }_{2}}\pi +{{N}_{x}}{{N}_{r}}{{N}_{c}}{{\log }_{2}}e\\
		&+{{N}_{r}}{{\log }_{2}}\left[ \det (\mathbf{X}{{\Sigma }_{\mathbf{G}}}{{\mathbf{X}}^{H}}+\sigma _{n}^{2}{{\mathbf{I}}_{{{N}_{x}}{{N}_{c}}}}) \right].\label{eq23}
	\end{split}
\end{equation}	

Similarly, the noise entropy can be obtained as
\begin{equation}
	\begin{split}
		h({{\mathbf{W}}_{\text{rad}}})=&{{N}_{x}}{{N}_{c}}{{N}_{r}}{{\log }_{2}}\pi +{{N}_{x}}{{N}_{c}}{{N}_{r}}{{\log }_{2}}e\\
		&+{{N}_{r}}{{\log }_{2}}\left[ \det (\sigma _{n}^{2}{{\mathbf{I}}_{{{N}_{x}}{{N}_{c}}}}) \right].
	\end{split}\label{eq24}
\end{equation}	
Finally, the sensing MI in (\ref{eq19}) is obtained by substituting (\ref{eq23}) and (\ref{eq24}) into (\ref{eq18}).
\end{proof}

The sensing MI derived above can be utilized for performance analysis and waveform design, revealing how much reflection from unknown targets can be captured for sensing. Note that the result in (\ref{eq19}) considers correlation in frequency doimain. Thus, we discuss the effect of subcarrier correlation as follows.
% \textit{Theorem 1:} According to the MIMO-OFDM signal model in (\ref{eq14}), the sensing MI of {\color{blue}the} MIMO-OFDM ISAC system can be expressed as
% \begin{equation}
% 	\begin{split}
% 		&I(\mathbf{G};{{\mathbf{Y}}_{\text{rad}}}|\mathbf{X}) \\
% 		&={{N}_{r}}{{\log }_{2}}\left[ {{(\sigma _{n}^{2})}^{-{{N}_{x}}{{N}_{c}}}}\det (\mathbf{X}{{\Sigma }_{\mathbf{G}}}{{\mathbf{X}}^{H}}+\sigma _{n}^{2}{{\mathbf{I}}_{{{N}_{x}}{{N}_{c}}}}) \right].\label{eq19}
% 	\end{split}
% \end{equation}	

% \textit{Proof:} 

\begin{corollary}
	If the sensing channel correlation of different subcarriers is not considered in (\ref{eq19}), the sensing MI of the MIMO-OFDM ISAC system is
	\begin{equation}
	\begin{split}
	&I(\mathbf{G'};{{\mathbf{Y}}_{\text{rad}}}|\mathbf{X}) \\
	&={{N}_{r}}{{\log }_{2}}\left[ {{(\sigma _{n}^{2})}^{-{{N}_{x}}{{N}_{c}}}}\det (\mathbf{X}{{\Sigma }_{\mathbf{G'}}}{{\mathbf{X}}^{H}}+\sigma _{n}^{2}{{\mathbf{I}}_{{{N}_{x}}{{N}_{c}}}}) \right],\label{eq19n}
	\end{split}
	\end{equation}
	where $\mathbf{G'}$ is the sensing channel without considering the correlation of different subcarriers,  and the covariance matrix of $\mathbf{G'}$ is 
\begin{equation}\label{eq15n}
%\begin{small}	
\begin{aligned}
 {\Sigma }_{\mathbf{G'}} 
 =\!\!\left( \begin{matrix}
E[\mathbf{G}(0){{\mathbf{G}}^{H}}(0)] &\!\!\! \ldots  \!\!\!& 0  \\
\vdots\!\!\!  & \ddots\!\!\!  & \vdots\!\!\!   \\
0 \!\!\!&\!\!\! \ldots\!\!\!  & E[\mathbf{G}({{N}_{c}}-1){{\mathbf{G}}^{H}}({{N}_{c}}-1)]  \\
\end{matrix} \right)\!\!. \\ 
\end{aligned}
%\end{small}
\end{equation}
\end{corollary}
\begin{proof}
	According to (\ref{eq15}), if the sensing channel correlation of different subcarriers is not considered, $E\left[ \mathbf{G}({{p}_{1}}){{\mathbf{G}}^{H}}({{p}_{2}}) \right]=\mathbf 0$, if $p_1 \ne p_2$. Thus, the channel correlation matrix is (\ref{eq15n}) and the sensing MI without considering the channel correlation of different subcarriers is obtained as (\ref{eq19n}).
\end{proof}

According to the above theoretical analysis, it can be found that ignoring the subcarrier correlation of the sensing channel will affect the performance of sensing, since the actual sensing channel still has the subcarrier correlation.

\begin{corollary}
	If the sensing channel is independent in both the spatial and frequency domains, the sensing MI of the MIMO-OFDM system is 
		\begin{equation}
	\begin{split}
	&I(\mathbf{G''};{{\mathbf{Y}}_{\text{rad}}}|\mathbf{X}) \\
	&={{N}_{r}}{{\log }_{2}}\left[ {{(\sigma _{n}^{2})}^{-{{N}_{x}}{{N}_{c}}}}\det (\mathbf{X}{{\mathbf{X}}^{H}}+\sigma _{n}^{2}{{\mathbf{I}}_{{{N}_{x}}{{N}_{c}}}}) \right],\label{eq19n2}
	\end{split}
	\end{equation}
	where $\mathbf{G''}$ is the sensing channel without considering the correlation in both the spatial and frequency domains.
\end{corollary}
\begin{proof}
	According to (\ref{eq17}), if the spatial correlation is not considered, the sensing channel correlation matrix is an identity matrix. (\ref{eq19n2}) is obtained by substituting the identity matrix into (\ref{eq19n}).
\end{proof}

The result in (\ref{eq19n2}) shows that without considering correlation in spatial and frequency domains, the sensing MI has nothing to do with the exact MIMO-OFDM channel and degenerates into the problem of single-carrier and single-antenna.

\subsection{Communication MI}
Communication MI is defined as the MI between the transmitted signal and the received signal given the knowledge of CSI, as given by
\begin{equation}
	I(\mathbf{X};{{\mathbf{Y}}_{\text{com}}}|\mathbf{H})=h({{\mathbf{Y}}_{\text{com}}}|\mathbf{H})-h({{\mathbf{W}}_{\text{com}}}).\label{eq25}
\end{equation}

\begin{theorem}
    According to the signal model in (\ref{eq6}), the communication MI of the MIMO-OFDM ISAC system is obtained as
\begin{equation}
	\begin{split}
		&I(\mathbf{X};{{\mathbf{Y}}_{\text{com}}}|\mathbf{H})\\
		&={{N}_{x}}{{\log }_{2}}\left[ {{(\sigma _{n}^{2})}^{-{{N}_{r}}{{N}_{c}}}}\det ({{\mathbf{H}}^{H}}{{\Sigma }_{\mathbf{X}}}\mathbf{H}+\sigma _{n}^{2}{{\mathbf{I}}_{{{N}_{r}}{{N}_{c}}}}) \right].\label{eq26}
	\end{split}
\end{equation}	
\end{theorem}

\begin{proof}
	With perfect CSI, the PDF of received data ${\mathbf{Y}}_{\text{com}}$ is
\begin{equation}\label{eq28}
\begin{small}
\begin{aligned}
&p({{\mathbf{Y}}_{\text{com}}}|\mathbf{H})=\prod\limits_{j=1}^{{{N}_{x}}{{N}_{c}}}{p(\mathbf{y}_{\text{com}}^{j}|\mathbf{H})} \\
%&=\prod\limits_{j=1}^{{{N}_{x}}}{\frac{1}{{{\pi }^{{{N}_{r}}}}\left| \Lambda  \right|}\exp \left( -{{(\mathbf{y}_{\text{com},p}^{j})}^{H}}{{\Lambda }^{-1}}\mathbf{y}_{\text{com},p}^{j} \right)}.
&=\prod\limits_{j=1}^{{{N}_{x}}{{N}_{c}}}\frac{1}{{{\pi }^{{{N}_{r}}}}{{\det }}(\mathbf{H}^{H}{E\{{{(\mathbf{X}^{j})}^{H}}\mathbf{X}^{j}\}}{{\mathbf{H}}}+\sigma _{n}^{2}{{\mathbf{I}}_{{{N}_{r}}}})}\\
&\times\!\exp\!\!\left(\!\!-\!\left( {{{({{\mathbf{H}}^{H}}{E\{{{(\mathbf{X}^{j})}^{H}}\mathbf{X}^{j}\}}\mathbf{H}\!+\!\sigma _{n}^{2}{{\mathbf{I}}_{{{N}_{r}}}})}^{-1}}\mathbf{y}_{\text{com}}^{j}{{(\mathbf{y}_{\text{com}}^{j})}^{H}}}\!\! \right)\!\! \right)\!\!.
\end{aligned}
\end{small}
\end{equation}

Since $\mathbf X$ is a block diagonal matrix, (\ref{eq28}) can be transformed into the product of the PDF at different subcarriers, which is rewritten as
\begin{equation}\label{eq29}
\begin{small}
\begin{aligned}
&p({{\mathbf{Y}}_{\text{com}}}|\mathbf{H})\\
%&=\prod\limits_{j=1}^{{{N}_{x}}}{\frac{1}{{{\pi }^{{{N}_{r}}}}\left| \Lambda  \right|}\exp \left( -{{(\mathbf{y}_{\text{com},p}^{j})}^{H}}{{\Lambda }^{-1}}\mathbf{y}_{\text{com},p}^{j} \right)}.
&=\prod\limits_{j=1}^{{{N}_{c}}}\frac{1}{{{\pi }^{{{N}_{x}}{{N}_{r}}}}{{\det }^{{{N}_{x}}}}(\mathbf{H}(p){{\Sigma }_{\mathbf{X}(p)}}{{\mathbf{H}}^{H}}(p)+\sigma _{n}^{2}{{\mathbf{I}}_{{{N}_{r}}}})}\\
&\times\!\exp\!\!\left(\!\!-\text{tr}\!\left( \sum\limits_{j=1}^{{{N}_{x}}}{{{({{\mathbf{H}}^{H}}(p){{\Sigma }_{\mathbf{X}(p)}}\mathbf{H}(p)\!+\!\sigma _{n}^{2}{{\mathbf{I}}_{{{N}_{r}}}})}^{-1}}\mathbf{y}_{\text{com},p}^{j}{{(\mathbf{y}_{\text{com},p}^{j})}^{H}}}\!\! \right)\!\! \right)\!\!,
\end{aligned}
\end{small}
\end{equation}
where 
\begin{equation}
E\{{{\mathbf{X}}^{H}}(p)\mathbf{X}(p)\}/{{N}_{x}}=E\{{{(\mathbf{X}{{(p)}^{j}})}^{H}}\mathbf{X}{{(p)}^{j}}\}={{\Sigma }_{\mathbf{X}(p)}}.
\end{equation}

In (\ref{eq29}), it is proved that the received data at each subcarrier is independent of each other with the knowledge of perfect CSI. The correlation matrix of $\mathbf{y}_{\text{com},p}^{j}$ is derived as 
\begin{equation}\label{eq30}
\begin{split}
E\{\mathbf{y}_{\text{com},p}^{j}{{(\mathbf{y}_{\text{com},p}^{j})}^{H}}\}
=\mathbf{H}{{(p)}^{H}}E\{{\Sigma }_{\mathbf{X}(p)}\}\mathbf{H}(p)+\sigma _{n}^{2}{{\mathbf{I}}_{{{N}_{r}}}}.
\end{split}
\end{equation}
The entropy matrix $h({{\mathbf{Y}}_{\text{com}}}|\mathbf{H})$ is obtained as
\begin{equation}\label{eq31}
	\begin{split}
		h({{\mathbf{Y}}_{\text{com}}}|\mathbf{H})=&\sum\limits_{p=0}^{{{N}_{c}}-1}\left\lbrace {{N}_{x}}{{N}_{r}}{{\log }_{2}}\pi +{{N}_{x}}{{N}_{r}}{{\log }_{2}}e\right.\\ 
		&\left.+{{N}_{x}}{{\log }_{2}}\left[ \det ({{\mathbf{H}}^{H}}(p){{\Sigma }_{\mathbf{X}(p)}}\mathbf{H}(p)
		+\sigma _{n}^{2}{{\mathbf{I}}_{{{N}_{r}}}})\! \right]\right\rbrace\!.
	\end{split}
\end{equation}
The noise entropy can be obtained as
\begin{equation}\label{eq32}
	\begin{split}
		h({{\mathbf{W}}_{\text{com}}})=&{{N}_{x}}{{N}_{r}}{{N}_{c}}{{\log }_{2}}\pi +{{N}_{x}}{{N}_{r}}{{N}_{c}}{{\log }_{2}}e\\
		&+{{N}_{x}}{{N}_{c}}{{\log }_{2}}\left[ \det (\sigma _{n}^{2}{{\mathbf{I}}_{{{N}_{r}}}}) \right].
	\end{split}
\end{equation}

Substituting (\ref{eq31}) and (\ref{eq32}) into (\ref{eq25}), the communication MI is obtained as
\begin{equation}\label{eq34}
	\begin{split}
		&I(\mathbf{X};{{\mathbf{Y}}_{\text{com}}}|\mathbf{H})\\
		&=\!\!{{N}_{x}}\!\!\sum\limits_{p=0}^{{{N}_{c}}-1}{{{\log }_{2}}\left[ {{(\sigma _{n}^{2})}^{-{{N}_{r}}}}\det ({{\mathbf{H}}^{H}}(p){{\Sigma }_{\mathbf{X}(p)}}\mathbf{H}(p)+\sigma _{n}^{2}{{\mathbf{I}}_{{{N}_{r}}}}) \right]}.
	\end{split}
\end{equation}

Due to the summation, the communication MI (\ref{eq34}) cannot directly represent the communication model in (\ref{eq6}). Since the determinant value of the block diagonal matrix can be transformed into
\begin{equation}\label{eq35}
	\left| \begin{matrix}
		\mathbf{A} & \mathbf{0}  \\
		\mathbf{0} & \mathbf{B}  \\
	\end{matrix} \right|=\det (\mathbf{A})\det (\mathbf{B}),
\end{equation}
which can also be extended to a multi-block diagonal matrix, so that we have
\begin{equation}\label{eq36}
	\begin{small}	
		\begin{aligned}
			& I(\mathbf{X};{{\mathbf{Y}}_{\text{com}}}|\mathbf{H})\\
			&=\!\!{{N}_{x}}{{\log }_{2}}\!\left[ {{(\sigma _{n}^{2})}^{-{{N}_{r}}{{N}_{c}}}}\prod\nolimits_{p=0}^{{{N}_{c}}-1}{\det ({{\mathbf{H}}^{H}}(p){{\Sigma }_{\mathbf{X}(p)}}\mathbf{H}(p)+\sigma _{n}^{2}{{\mathbf{I}}_{{{N}_{r}}}})} \right] \\ 
			& ={{N}_{x}}{{\log }_{2}}\left[ {{(\sigma _{n}^{2})}^{-{{N}_{r}}{{N}_{c}}}}\det ({{\mathbf{H}}^{H}}\text{diag}\{{{\Sigma }_{\mathbf{X}(p)}}\}\mathbf{H}+\sigma _{n}^{2}{{\mathbf{I}}_{{{N}_{r}}{{N}_{c}}}}) \right], \\ 
		\end{aligned}
	\end{small}
\end{equation}
where 
\begin{equation}
\mathbf{H}=\text{diag}\{\mathbf{H}(0),\cdots ,\mathbf{H}({{N}_{c}}-1)\}\in {{\mathbb{C}}^{{{N}_{t}}{{N}_{c}}\times {{N}_{r}}{{N}_{c}}}}
\end{equation}
is the transformation of $\mathbf H$ in (\ref{eqh}).

Since $\mathbf{X}=\text{diag}\{\mathbf{X}(0),\cdots ,\mathbf{X}({{N}_{c}}-1)\}\in {{\mathbb{C}}^{{{N}_{c}}{{N}_{x}}\times {{N}_{c}}{{N}_{t}}}}$ is a block diagonal matrix, the symbols of each subcarrier are distributed independently, such that $\text{diag}\{{{\Sigma }_{\mathbf{X}(p)}}\}={{\Sigma }_{\mathbf{X}}}$, where $E\{{{\mathbf{X}}^{H}}\mathbf{X}\}/{{N}_{x}}={{\Sigma }_{\mathbf{X}}}$. Finally, the communication MI is obtained as (\ref{eq26}).
\end{proof}

Note that different from the sensing MI, the correlation of subcarriers in communication MI is not considered, since the communication channel matrix can be transformed into block diagonal form. As the sensing channel needs to be sensed and the derivation is based on the known transmitted signals, the sensing MI can't be transformed into the sum of the MI at different subcarriers.
%\begin{remark}
%Note that the communication MI of the MIMO-OFDM ISAC system derived in (\ref{eq26}) can also be applied to the case without considering the correlation between different subcarriers in communication channels, since the communication channel matrix can be transformed into block diagonal form.
%\end{remark}

% \textit{Proof:} 
% $\hfill\blacksquare$ 

\section{MIMO-OFDM ISAC Waveform Design}
% We will discuss the optimization of transmitted signals in this chapter. 
In this section, the waveform of the MIMO-OFDM ISAC system is optimized.
Since the channel is frequency-selective, the optimization according to the channel characteristics affects the sensing and communication MI. We propose an optimization scheme based on the weighted sum of communication and sensing MIs, and provide two special cases, including maximizing sensing MI and maximizing communication MI.

\subsection{Weighted Sum of Sensing and Communication }
In the MIMO-OFDM ISAC system, maximizing the performance of sensing and communication is essential for the ISAC waveform design. Since the MIs of communication and sensing are different in the orders of magnitude, in order to comprehensively optimize both sensing and communication performance, the weighted sum of the normalized sensing and communication MIs is applied in the optimization, as given by
\begin{equation}\label{eq54}
	\begin{aligned}
		& \underset{\mathbf{X}}{\mathop{\max }}\,\text{    }{{F}_{\omega }}=\frac{{{\omega }_{r}}}{{{F}_{r}}}I(\mathbf{G};{{\mathbf{Y}}_{\text{rad}}}|\mathbf{X})+\frac{{{\omega }_{c}}}{{{F}_{c}}}I(\mathbf{X};{{\mathbf{Y}}_{\text{com}}}|\mathbf{H}), \\ 
		& \text{ s}\text{.t}\text{.      tr}\left( \mathbf{X}{{\mathbf{X}}^{H}} \right)\le E, 
	\end{aligned}
\end{equation}
where ${{\omega }_{r}}$ and ${{\omega }_{c}}$ are the weighting coefficients of sensing and communication, respectively. ${{F}_{r}}$ and ${{F}_{c}}$ are the maximum sensing MI and communication MI, respectively. $E$ represents the average transmit power constraint.

Since ${{\Sigma }_{\mathbf{G}}}$ is a Hermitian matrix, we use the singular value decomposition (SVD) to decompose ${{\Sigma }_{\mathbf{G}}}$ as ${{\Sigma }_{\mathbf{G}}}={{\mathbf{U}}_{\mathbf{G}}}{{\mathbf{\Lambda} }_{\mathbf{G}}}\mathbf{U}_{\mathbf{G}}^{H}$, where ${{\mathbf{U}}_{\mathbf{G}}}$ is a unitary matrix, and ${{\mathbf{\Lambda} }_{\mathbf{G}}}$ is a diagonal matrix satisfying
\begin{equation}\label{eq38}
{{\mathbf{\Lambda} }_{\mathbf{G}}}=\text{diag}\{{{\lambda }_{11}},\cdots ,{{\lambda }_{{{N}_{t}}{{N}_{c}}}}\}.
\end{equation}
Note that ${{\Sigma }_{\mathbf{G}}}$ can also be described in the time-domain as
\begin{equation}\label{eq38n}
{{\Sigma }_{\mathbf{G}}}=\mathbf{\Omega}_r{{\Sigma }_{\mathbf{g}}}{\mathbf{\Omega}_r}^H,
\end{equation}
where ${{\Sigma }_{\mathbf{g}}}=E[{\mathbf{g}\mathbf{g}^H}]/N_r$, satisfying $\mathbf G=\mathbf \Omega_r \mathbf g$, and $\mathbf{\Omega}_r={{[{{\mathbf{\Omega}_r }^{T}}(0),{{\mathbf\Omega_r }^{T}}(1),\ldots ,{{\mathbf\Omega_r }^{T}}({{N}_{c}}-1)]}^{T}}\in {{\mathbb{C}}^{{{N}_{c}}{{N}_{t}}\times {{N}_{t}}(L+1)}}$, ${\mathbf\Omega_r}(p)={{\mathbf{I}}_{{{N}_{t}}}}\otimes \omega_r (p)$. 
(\ref{eq38n}) indicates that the channel correlation matrix comprises both spatial correlation and subcarrier correlation.

The SVD of the communication channel correlation matrix is different from that of the sensing channel correlation matrix because of the perfect CSI, so that we have $\mathbf{H}{{\mathbf{H}}^{H}}={{\mathbf{U}}_{\mathbf{H}}}{{\mathbf{\Lambda} }_{\mathbf{H}}}\mathbf{U}_{\mathbf{H}}^{H}$, where ${{\Lambda }_{\mathbf{H}}}$ is a diagonal matrix satisfying
\begin{equation}\label{eq47}
{{\mathbf{\Lambda} }_{\mathbf{H}}}=\text{diag}\{{{\mu }_{11}},\cdots ,{{\mu }_{{{N}_{t}}{{N}_{c}}}}\}.
\end{equation}
We further extend the expression of $\mathbf H {\mathbf H}^H$ to the time domain, as given by
\begin{equation}\label{eq47n}
{\mathbf H {\mathbf H}^H}=\text{diag}\left\lbrace\mathbf\Omega_c(p){{\Sigma }_{\mathbf{h}}} \mathbf\Omega_c^H(p)\right\rbrace ,
\end{equation}
where ${{\Sigma }_{\mathbf{h}}}={\mathbf{h}\mathbf{h}^H}$, satisfying $\mathbf H(p)=\mathbf \Omega_c(p) \mathbf h$ with $\mathbf \Omega_c(p)={{I}_{{{N}_{t}}}}\otimes\omega_c (p)$ and $\omega_c (p)=\left[ 1,{{e}^{-j2\pi p/{{N}_{c}}}},\right.$
$\left.\ldots ,{{e}^{-j2\pi pL_c/{{N}_{c}}}} \right]^T$.

Given that power allocation is a fundamental aspect of waveform design, the optimal waveform in this paper refers to the optimal power allocation \cite{zhou}\cite{caox}. Through the above derivation, the waveform optimization in this section is transformed into the power allocation with respect to the eigenvalues of the signal correlation matrix \cite{dong}. The optimal ISAC waveform design maximizing the weighted sum of communication and sensing MIs is revealed in Theorem \ref{theo-2}.

%Through the above derivation, the optimal ISAC waveform design maximizing the weighted sum of communication and sensing MIs is revealed in Theorem \ref{theo-2}.

\begin{theorem}\label{theo-2}
    Define $\mathbf{\Xi }=\mathbf{U}_{\mathbf{G}}^{H}{{\mathbf{X}}^{H}}\mathbf{X}{{\mathbf{U}}_{\mathbf{G}}}=\mathbf{U}_{\mathbf{H}}^{H}{{\Sigma }_{\mathbf{X}}}{{\mathbf{U}}_{\mathbf{H}}}\buildrel \Delta \over = \left[ {{\xi_{ij}}} \right]$, the optimization problem of the MIMO-OFDM ISAC waveform design in (\ref{eq54}) can be transformed into
\begin{equation}\label{eq55}
	\begin{aligned}
		&\underset{\Xi }{\mathop{\max }}\,&{{F}_{\omega }}(\mathbf{\Xi})&=\sum\limits_{i=1}^{{{N}_{t}}{{N}_{c}}}\left\{ \frac{{{\omega }_{r}}}{{{F}_{r}}}{{N}_{r}}{{\log }_{2}}({{\lambda }_{ii}}{{\xi }_{ii}}/\sigma _{n}^{2}+1)\right.\\
		& & &\left.+\frac{1-{{\omega }_{r}}}{{{F}_{c}}}{{N}_{x}}{{\log }_{2}}({{\mu }_{ii}}{{\xi }_{ii}}/\sigma _{n}^{2}+1) \right\}, \\ 
		& \text{ s}\text{.t}\text{.}&\text{tr}(\mathbf{\Xi })&\le E,\text{ }{{\xi }_{ii}}\ge 0,\text{ }1\le i\le {{N}_{c}}{{N}_{t}}. 
	\end{aligned}
\end{equation}
The optimal allocation scheme is 
% \begin{strip}
\begin{equation}\label{eq56}
		{{\xi }_{ii}}\!\!=\!\!\left\{ \begin{aligned}
		& \frac{1}{2}\!\left[\!\frac{1}{\zeta }(\varepsilon +\eta )\!-\!(\frac{1}{{{\nu }_{i}}}\!+\!\frac{1}{{{\varphi }_{i}}})+\right.\\
		&\left.\;\;\sqrt{{{[(\frac{1}{{{\nu }_{i}}}\!-\!\frac{1}{{{\varphi }_{i}}})\!+\!\frac{1}{\gamma }(\eta \!-\!\varepsilon )]}^{2}}\!+\!4\frac{\varepsilon \eta }{{{\gamma }^{2}}}} \right]^{+}\!\!\!,
		&{{\nu }_{i}}\ne 0, {{\varphi }_{i}}\ne 0 \\ 
		&0,&{{\nu }_{i}}=0, {{\varphi }_{i}}=0 \\ 
		\end{aligned} \right.,
		\end{equation}
		\begin{equation}\label{eq57}
		\sum\limits_{i=1}^{{{N}_{t}}{{N}_{c}}}{{{\xi }_{ii}}=E},
\end{equation}

% \end{strip}
where ${{v}_{i}}=\frac{{{\lambda }_{ii}}}{\sigma _{n}^{2}}$, ${{\varphi }_{i}}=\frac{{{\mu }_{ii}}}{\sigma _{n}^{2}}$, $\varepsilon =\frac{{{\omega }_{r}}{{N}_{r}}}{{{F}_{r}}\ln 2}$, $\eta =\frac{(1-{{\omega }_{r}}){{N}_{x}}}{{{F}_{c}}\ln 2}$.
\end{theorem}

\begin{proof}
	According to the Sylvester's determinant equation $\det (\mathbf{AB}+\sigma _{n}^{2}{{\mathbf{I}}_{N}})={{(\sigma _{n}^{2})}^{n-m}}\det (\mathbf{BA}+\sigma _{n}^{2}{{\mathbf{I}}_{M}})$  \cite{b13}, (\ref{eq25}) and (\ref{eq38}) can be transformed into
	\begin{equation}\label{eq42}
	I(\mathbf{G};{{\mathbf{Y}}_{\text{rad}}}|\mathbf{X})={{N}_{r}}{{\log }_{2}}\left[ {{(\sigma _{n}^{2})}^{-{{N}_{t}}{{N}_{c}}}}\det ({{\mathbf{\Lambda} }_{\mathbf{G}}}\mathbf{\Xi}+\sigma _{n}^{2}{{\mathbf{I}}_{{{N}_{t}}{{N}_{c}}}}) \right]\!.
	\end{equation}
	\begin{equation}\label{eq51}
	I(\mathbf{X};{{\mathbf{Y}}_{\text{com}}}|\mathbf{H})={{N}_{x}}{{\log }_{2}}\!\left[ {{(\sigma _{n}^{2})}^{-{{N}_{t}}{{N}_{c}}}}\det (\mathbf{\Xi}{{\mathbf{\Lambda} }_{\mathbf{H}}}+\sigma _{n}^{2}{{\mathbf{I}}_{{{N}_{t}}{{N}_{c}}}}) \right]\!\!.
	\end{equation}	 
	
	According to the Hadamard's inequality, $\det (\mathbf{\Xi })\le \prod\nolimits_{i=1}^{{{N}_{t}}{{N}_{c}}}{{{\xi }_{ii}}}$. The rest part is bounded by 
	\begin{equation}
	\det ({{\mathbf{\Lambda} }_{\mathbf{G}}}+\sigma _{n}^{2}{{\mathbf{\Xi}}^{-1}})\le \prod\nolimits_{i=1}^{{{N}_{t}}{{N}_{c}}}{({{\lambda }_{ii}}+\sigma _{n}^{2}/{{\xi}_{ii}})}. 
	\end{equation}
	\begin{equation}
	\det ({{\mathbf{\Lambda} }_{\mathbf{H}}}+\sigma _{n}^{2}{{\mathbf{\Xi}}^{-1}})\le \prod\nolimits_{i=1}^{{{N}_{t}}{{N}_{c}}}{({{\mu }_{ii}}+\sigma _{n}^{2}/{{\xi}_{ii}})}. 
	\end{equation}

As ${{\mathbf{U}}_{\mathbf{G}}}$ and ${{\mathbf{U}}_{\mathbf{H}}}$ are unitary matrices, we have
\begin{equation}
\begin{aligned}
\text{tr}(\mathbf{\Xi})  = \text{tr}(\mathbf{U}_{\mathbf{H}}^{H}{{\Sigma }_{\mathbf{X}}}{{\mathbf{U}}_{\mathbf{H}}}) 
=\text{tr}(\mathbf{U}_{\mathbf{G}}^{H}{{\mathbf{X}}^{H}}\mathbf{X}{{\mathbf{U}}_{\mathbf{G}}})=\text{tr}({{\Sigma }_{\mathbf{X}}})\le E.
\end{aligned}
\end{equation}

Through the above derivation, the optimization problem can be simplified to (\ref{eq55}). Obviously, the objective is the weighted sum of two concave functions, and the concave function maximization problem can be converted into a convex function minimization problem. Thus, the problem is resolved by using the KKT condition \cite{boyd}, which achieves the maximum weighted MI.

The Lagrange function is  
\begin{equation}\label{eq59}
	\begin{aligned}
	L(\mathbf{\Xi })&=-\sum\limits_{i=1}^{{{N}_{t}}{{N}_{c}}}\left\{ \frac{{{\omega }_{r}}}{{{F}_{r}}}{{N}_{r}}{{\log }_{2}}({{\lambda }_{ii}}{{\xi }_{ii}}/\sigma _{n}^{2}+1)\right.\\
	&\left.+\frac{1-{{\omega }_{r}}}{{{F}_{c}}}{{N}_{x}}{{\log }_{2}}({{\mu }_{ii}}{{\xi }_{ii}}/\sigma _{n}^{2}+1) \right\} \\ 
	& +\gamma (\sum\limits_{i=1}^{{{N}_{t}}{{N}_{c}}}{{{\xi }_{ii}}}-E)+\sum\limits_{i=1}^{{{N}_{t}}{{N}_{c}}}{{{\gamma }_{i}}(-{{\xi }_{ii}})}. \\ 
	\end{aligned}
\end{equation}

Setting the partial derivative ${{\nabla }_{{{\xi }_{ii}}}}L(\mathbf{\Xi })$ of ${{\xi }_{ii}}$ to zero, and supplementing the relaxation complementarity condition, the inequality constraint and the KKT multiplier constraint, we have
\begin{align}
	\gamma -{{\gamma }_{i}}&=\frac{{{\omega }_{r}}{{N}_{r}}}{{{F}_{r}}\ln 2}\frac{{{\lambda }_{ii}}}{\sigma _{n}^{2}+{{\lambda }_{ii}}{{\xi }_{ii}}}+\frac{(1-{{\omega }_{r}}){{N}_{x}}}{{{F}_{c}}\ln 2}\frac{{{\mu }_{ii}}}{\sigma _{n}^{2}+{{\mu }_{ii}}{{\xi }_{ii}}}\notag \\
	&=\varepsilon \frac{{{\nu }_{i}}}{1+{{\nu }_{i}}{{\xi }_{ii}}}+\eta \frac{{{\varphi }_{i}}}{1+{{\varphi }_{i}}{{\xi }_{ii}}},\label{eq60}
\end{align}
\begin{equation}\label{eq61}
	\gamma \left( \sum\limits_{i=1}^{{{N}_{t}}{{N}_{c}}}{{{\xi }_{ii}}-E} \right)=0,
\end{equation}
\begin{equation}\label{eq62}
	{{\gamma }_{i}}{{\xi }_{ii}}=0,
\end{equation}
\begin{equation}\label{eq63}
	\gamma \ge 0,{{\gamma }_{i}}\ge 0,1\le i\le {{N}_{t}}{{N}_{c}}.
\end{equation}

Solving (\ref{eq60})--(\ref{eq62}), if ${{\xi }_{ii}}\ne 0$, ${{\gamma }_{i}}$ is zero and the optimal solution can be obtained as (\ref{eq56}) and (\ref{eq57}).%$\hfill\blacksquare$
\end{proof}

Through the binary search method, the optimal $\gamma $ can be searched in $0<1/\gamma <1/\min \{\varepsilon /(1/{{\nu }_{i}}+E)+\eta /(1/{{\varphi }_{i}}+E)\}$, so that the optimal solution can be found. 

Finally, the sensing MI and communication MI are obtained as
\begin{equation}\label{eq64}
	I(\mathbf{G};{{\mathbf{Y}}_{\text{rad}}}|\mathbf{X})\!=\!{{N}_{r}}{{\log }_{2}}\left[ {{(\sigma _{n}^{2})}^{-{{N}_{t}}{{N}_{c}}}}\det ({{\Lambda }_{\mathbf{G}}}\mathbf{\Xi }+\sigma _{n}^{2}{{\mathbf{I}}_{{{N}_{t}}{{N}_{c}}}}) \right],
\end{equation}
\begin{equation}\label{eq65}
	I(\mathbf{X};{{\mathbf{Y}}_{\text{com}}}|\mathbf{H})\!=\!{{N}_{x}}{{\log }_{2}}\left[ {{(\sigma _{n}^{2})}^{-{{N}_{t}}{{N}_{c}}}}\!\det (\mathbf{\Xi }{{\Lambda }_{\mathbf{H}}}+\sigma _{n}^{2}{{\mathbf{I}}_{{{N}_{t}}{{N}_{c}}}}) \right]\!.\!
\end{equation}

Therefore, the weighted MI is also obtained as
\begin{equation}\label{eq66}
	{F_\omega}\!\!=\!\!\sum\limits_{i=1}^{{{N}_{t}}{{N}_{c}}}\!\!{\left\{\! \frac{{{N}_{r}}{{\log }_{2}}({{\lambda }_{ii}}{{\xi }_{ii}}/\sigma _{n}^{2}+1)}{{{F}_{r}}}\!\!+\!\!\frac{{{N}_{x}}{{\log }_{2}}({{\mu }_{ii}}{{\xi }_{ii}}/\sigma _{n}^{2}+1)}{{{F}_{c}}} \!\right\}}.
\end{equation}

\begin{remark}
If the channel correlation between different subcarriers is not considered, similar to (\ref{eq64}), the sensing MI is obtained as
\begin{equation}\label{eq67}
I(\mathbf{G'};{{\mathbf{Y}}_{\text{rad}}}|\mathbf{X})\!\!=\!\!{{N}_{r}}{{\log }_{2}}\left[ \!{{(\sigma _{n}^{2})}^{-{{N}_{t}}{{N}_{c}}}}\det ({{\Lambda }_{\mathbf{G'}}}\mathbf{\Xi }+\sigma _{n}^{2}{{\mathbf{I}}_{{{N}_{t}}{{N}_{c}}}}) \!\right]\!,
\end{equation}
and the communication MI is (\ref{eq65}). The weighted MI is
\begin{equation}\label{eq69}
{F'_\omega}\!\!=\!\!\sum\limits_{i=1}^{{{N}_{t}}{{N}_{c}}}\!\!{\left\{\! \frac{{{N}_{r}}{{\log }_{2}}({{\lambda }'_{ii}}{{\xi }_{ii}'}/\sigma _{n}^{2}+1)}{{{F}_{r}}}\!\!+\!\!\frac{{{N}_{x}}{{\log }_{2}}({{\mu }_{ii}}{{\xi }_{ii}'}/\sigma _{n}^{2}+1)}{{{F}_{c}}} \!\right\}}.
\end{equation}
Due to the channel independence between different subcarriers, we have $I(\mathbf{G'};{{\mathbf{Y}}_{\text{rad}}}|\mathbf{X})>I(\mathbf{G};{{\mathbf{Y}}_{\text{rad}}}|\mathbf{X})$. Thus, ${F'_\omega}>{F_\omega}$. Although the results show a bigger upper bound without considering the channel correlation between different subcarriers, it is not practical to ignore the channel correlation, which may lose the characteristics of the channel, so that the optimal result is not obtained from the real channel and the performance of ISAC system is not the best.
\end{remark}
%\begin{remark}
	%If the channel is independent in both space and frequency domain, the MI for sensing and communication does not depend on the channel correlation matrix, indicating that the optimal allocation scheme would be equal allocation.
	%\begin{equation}
	%     {\xi ''}=\frac{E}{{N_t}{N_c}}.
	%\end{equation}
%\end{remark}

\subsection{Special Case 1: Maximizing sensing MI}
For sensing MI maximization, the transmitted signal should be adapted to the sensing channel, satisfying the average transmit power constraint. The optimization is formulated as
\begin{equation}\label{eq37}
\begin{aligned}
\begin{matrix}
\underset{\mathbf{X}}{\mathop{\max }}\,&I(\mathbf{G};{{\mathbf{Y}}_{\text{rad}}}|\mathbf{X}), \\ 
\text{  s}\text{.t}\text{.}&\text{tr}\left( \mathbf{X}{{\mathbf{X}}^{H}} \right)\le E. \\ 
\end{matrix}
\end{aligned}
\end{equation}

The ISAC waveform design maximizing the sensing MI is revealed in Corollary \ref{theo-3}.

\begin{corollary}\label{theo-3}
	Define matrix $\mathbf{Q}=\mathbf{U}_{\mathbf{G}}^{H}{{\mathbf{X}}^{H}}\mathbf{X}{{\mathbf{U}}_{\mathbf{G}}}\buildrel \Delta \over = \left[ {{q_{ij}}} \right]$. Then the optimization maximizing sensing MI can be transformed into
	\begin{equation}\label{eq39}
	\begin{aligned}
	&\underset{\mathbf{Q}}{\mathop{\max }}\,&{{F}_{r}}(\mathbf{Q})&={{N}_{r}}\sum\nolimits_{i=1}^{{{N}_{t}}{{N}_{c}}}{{{\log }_{2}}({{\lambda }_{ii}}{{q}_{ii}}/\sigma _{n}^{2}+1)}, \\ 
	&\text{s}\text{.t}\text{.}&\text{tr}(\mathbf{Q})&\le E,\text{ }{{q}_{ii}}\ge 0,\text{ }1\le i\le {{N}_{c}}{{N}_{t}}. \\ 
	\end{aligned}
	\end{equation}
	The optimal allocation scheme is
	\begin{equation}\label{eq40}
	{{q}_{ii}}=\left\{ \begin{aligned}
	& {{(-\frac{1}{\alpha \ln 2}-\frac{\sigma _{n}^{2}}{{{\lambda }_{ii}}})}^{+}}&{{\lambda }_{ii}}\ne 0 \\ 
	& 0&{{\lambda }_{ii}}=0 \\ 
	\end{aligned} \right.\text{ },\text{ }i=1,\cdots ,{{N}_{c}}{{N}_{t}},
	\end{equation}
	\begin{equation}\label{eq41}
	\sum\limits_{i=1}^{{{N}_{t}}{{N}_{c}}}{{{q}_{ii}}}-E=0,
	\end{equation}
	where ${{[x]}^{+}}=\max \{x,0\}$.
\end{corollary}

% \textit{Theorem 3:} Define matrix $\mathbf{Q}=\mathbf{U}_{\mathbf{G}}^{H}{{\mathbf{X}}^{H}}\mathbf{X}{{\mathbf{U}}_{\mathbf{G}}}\buildrel \Delta \over = \left[ {{q_{ij}}} \right]$. $E$ is the total average transmit power. Then the optimization problem maximizing sensing MI can be transformed into
% \begin{equation}\label{eq39}
% 	\begin{aligned}
% 		&\underset{\mathbf{Q}}{\mathop{\max }}\,&{{F}_{r}}(\mathbf{Q})&={{N}_{r}}\sum\nolimits_{i=1}^{{{N}_{t}}{{N}_{c}}}{{{\log }_{2}}({{\lambda }_{ii}}{{q}_{ii}}/\sigma _{n}^{2}+1)}, \\ 
% 		&\text{s}\text{.t}\text{.}&\text{tr}(\mathbf{Q})&\le E,\text{ }{{q}_{ii}}\ge 0,\text{ }1\le i\le {{N}_{c}}{{N}_{t}}. \\ 
% 	\end{aligned}
% \end{equation}

% \textnormal{The optimal power allocation scheme is}
% \begin{equation}\label{eq40}
% 	{{q}_{ii}}=\left\{ \begin{aligned}
% 		& {{(-\frac{1}{\alpha \ln 2}-\frac{\sigma _{n}^{2}}{{{\lambda }_{ii}}})}^{+}}&{{\lambda }_{ii}}\ne 0 \\ 
% 		& 0&{{\lambda }_{ii}}=0 \\ 
% 	\end{aligned} \right.\text{ },\text{ }i=1,\cdots ,{{N}_{c}}{{N}_{t}},
% \end{equation}
% \begin{equation}\label{eq41}
% 	\sum\limits_{i=1}^{{{N}_{t}}{{N}_{c}}}{{{q}_{ii}}}-E=0,
% \end{equation}
% where ${{[x]}^{+}}=\max \{x,0\}$.

\begin{proof}	
	To solve (\ref{eq39}), the Lagrange multiplier method is applied. The Lagrange function is
	\begin{equation}\label{eq44}
	L(\mathbf{Q})=\sum\limits_{i=1}^{{{N}_{t}}{{N}_{c}}}{{{\log }_{2}}({{\lambda }_{ii}}{{q}_{ii}}/\sigma _{n}^{2}+1)+}\alpha (\sum\limits_{i=1}^{{{N}_{t}}{{N}_{c}}}{{{q}_{ii}}}-E),	
	\end{equation}
	which transforms the problem (\ref{eq39}) into a Lagrange function extremum problem. Here, $\alpha $ is related to ${{q}_{ii}}$. By setting the partial derivatives ${{\nabla }_{{{q}_{ii}}}}L(\mathbf{Q})$ of ${{q}_{ii}}$ and the partial derivative ${{\nabla }_{\alpha }}L(\mathbf{Q})$ of $\alpha $ to zero, the final closed-form solutions are obtained, i.e., (\ref{eq40}) and (\ref{eq41}). %$\hfill\blacksquare$
\end{proof}

The diagonal elements of $\mathbf{Q}$ are non-negative real numbers, so that ${{\mathbf{Q}}^{1/2}}$ exists. Meanwhile, we define $\mathbf{Z=X}{{\mathbf{U}}_{\mathbf{G}}}$, which satisfies $\text{tr}(\mathbf{Q})=\text{tr}(\mathbf{U}_{\mathbf{G}}^{H}{{\mathbf{X}}^{H}}\mathbf{X}{{\mathbf{U}}_{\mathbf{G}}})=\operatorname{tr}({{\mathbf{Z}}^{H}}\mathbf{Z})$. For any matrix $\mathbf{\Phi }(p)\in {{\mathbb{C}}^{{{N}_{x}}\times {{N}_{t}}}}$ with orthogonal columns, which satisfies $\mathbf{\Phi }=\text{diag}\{\mathbf{\Phi }(p)\}$, we have $\mathbf{Z=\Phi }{{\mathbf{Q}}^{1/2}}$. Through the above derivation, the transmitted symbol can be obtained as $\mathbf{X}=\mathbf{\Phi }{{\mathbf{Q}}^{1/2}}\mathbf{U}_{\mathbf{G}}^{H}$. The optimal transmitted signal obtained by maximizing the sensing MI is substituted into (\ref{eq26}), then the communication MI is 
\begin{equation}\label{eq45}
\begin{aligned}
&I(\mathbf{X};{{\mathbf{Y}}_{\text{com}}}|\mathbf{H})\\
&={{N}_{x}}{{\log }_{2}}\left[ {{(\sigma _{n}^{2})}^{-{{N}_{r}}{{N}_{c}}}}\det ({{\mathbf{H}}^{H}}{{\mathbf{U}}_{\mathbf{G}}}\mathbf{QU}_{\mathbf{G}}^{H}\mathbf{H}+\sigma _{n}^{2}{{\mathbf{I}}_{{{N}_{r}}{{N}_{c}}}}) \right].  
\end{aligned}
\end{equation}

The matrix in (\ref{eq45}) cannot be transformed into a diagonal form, so that (\ref{eq45}) is not the maximum communication MI.

\subsection{Special Case 2: Maximizing Communication MI}

For communication MI maximization, the optimization is formulated as
\begin{equation}\label{eq46}
\begin{matrix}
\underset{\mathbf{X}}{\mathop{\max }}\, & I(\mathbf{X};{{\mathbf{Y}}_{\text{com}}}|\mathbf{H}),  \\
\text{s}\text{.t}\text{.} & \text{tr}\left( \mathbf{X}{{\mathbf{X}}^{H}} \right)\le E.  \\
\end{matrix}
\end{equation}

The ISAC waveform design maximizing the communication MI is revealed in Corollary \ref{theo-4}.

\begin{corollary}\label{theo-4}
	By defining $\mathbf{S}=\mathbf{U}_{\mathbf{H}}^{H}{{\Sigma }_{\mathbf{X}}}{{\mathbf{U}}_{\mathbf{H}}}\buildrel \Delta \over = \left[ {{s_{ij}}} \right]
	$, the optimization maximizing communication MI can be transformed into
	\begin{equation}\label{eq48}
	\begin{aligned}
	& \underset{\mathbf{S}}{\mathop{\max }}\,&{{F}_{r}}(\mathbf{S})&={{N}_{x}}\sum\nolimits_{i=1}^{{{N}_{t}}{{N}_{c}}}{{{\log }_{2}}({{\mu }_{ii}}{{s}_{ii}}/\sigma _{n}^{2}+1)}, \\ 
	& \text{ s}\text{.t}\text{.}&\text{tr}(\mathbf{S})&\le E,\text{ }{{s}_{ii}}\ge 0,\text{ }1\le i\le {{N}_{c}}{{N}_{t}}. \\ 
	\end{aligned}
	\end{equation}
	The optimal allocation scheme is
	\begin{equation}\label{eq49}
	{{s}_{ii}}=\left\{ \begin{aligned}
	& {{(-\frac{1}{\beta \ln 2}-\frac{\sigma _{n}^{2}}{{{\mu }_{ii}}})}^{+}}&{{\mu }_{ii}}\ne 0 \\ 
	& 0&{{\mu }_{ii}}=0 \\ 
	\end{aligned} \right.\text{ },\text{ }i=1,\cdots ,{{N}_{c}}{{N}_{t}},
	\end{equation}
	\begin{equation}\label{eq50}
	\sum\limits_{i=1}^{{{N}_{t}}{{N}_{c}}}{{{s}_{ii}}}-E=0,
	\end{equation}	
	\textnormal{where} ${{[x]}^{+}}=\max \{x,0\}$.  
\end{corollary}

\begin{proof}
	%Compared to the optimization problem in (\ref{eq39}) maximizing sensing MI, the objectives and constraints of the two optimization problems are similar, so that 
	The closed-form solutions can be obtained as (\ref{eq49}) and (\ref{eq50}) from the extreme value conditions, and $\beta $ is 
	the Lagrange multiplier in the following Lagrange function
	\begin{align}
	L(\mathbf{S})\nonumber
	=\sum\nolimits_{i=1}^{{{N}_{t}}{{N}_{c}}}{{{\log }_{2}}({{\mu }_{ii}}{{s}_{ii}}/\sigma _{n}^{2}+1)+}\beta (\sum\nolimits_{i=1}^{{{N}_{t}}{{N}_{c}}}{{{s}_{ii}}}-E).
	\end{align}
\end{proof}

The diagonal elements of $\mathbf{S}$ are non-negative real numbers, so that ${{\mathbf{S}}^{1/2}}$ exists. Define $\mathbf{\Upsilon} \mathbf{=X}{{\mathbf{U}}_{\mathbf{H}}}$ satisfying $\text{tr}(\mathbf{S})=\text{tr}(\mathbf{U}_{\mathbf{H}}^{H}{{\Sigma }_{\mathbf{X}}}{{\mathbf{U}}_{\mathbf{H}}})=\text{tr}({{\mathbf{\Upsilon} }^{H}}\mathbf{\Upsilon} )$. For any $\mathbf{\Psi }(p)\in {{\mathbb{C}}^{{{N}_{x}}\times {{N}_{t}}}}$ with orthogonal columns satisfying $\mathbf{\Psi }=\text{diag}\{\mathbf{\Psi }(p)\}$, we have $\mathbf{\Upsilon} =\mathbf{\Psi }{{\mathbf{S}}^{1/2}}$. Through the above derivation, we have $\mathbf{X}=\mathbf{\Psi }{{\mathbf{S}}^{1/2}}\mathbf{U}_{\mathbf{H}}^{H}$. 
The optimal transmitted waveform obtained by maximizing communication MI is substituted into (\ref{eq19}), and the sensing MI is
\begin{equation}\label{eq53}
\begin{aligned}
&I(\mathbf{G};{{\mathbf{Y}}_{\text{rad}}}|\mathbf{X})\\
&={{N}_{r}}{{\log }_{2}}\left[ {{(\sigma _{n}^{2})}^{-{{N}_{t}}{{N}_{c}}}}\det ({{\Sigma }_{\mathbf{G}}}{{\mathbf{U}}_{\mathbf{H}}}\mathbf{SU}_{\mathbf{H}}^{H}+\sigma _{n}^{2}{{\mathbf{I}}_{{{N}_{t}}{{N}_{c}}}}) \right].  
\end{aligned}
\end{equation}

Since the matrix in (\ref{eq53}) cannot be transformed into a diagonal form, (\ref{eq53}) is obviously not the maximum sensing MI.

\begin{figure}[!t]
	\centering
	\includegraphics[width=0.5\textwidth]{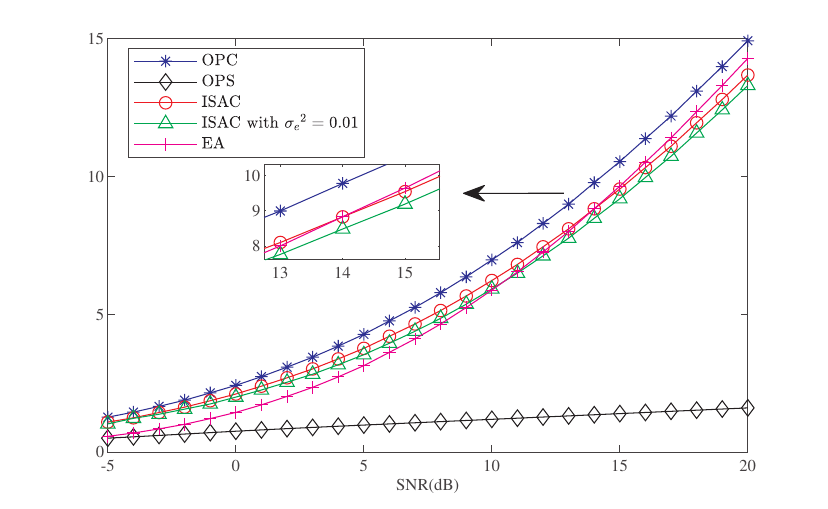}
	\caption{Spectral efficiency vs. SNR for different schemes.}
		\label{fig-2}
\end{figure}

\section{Simulation Results}

In this section, the above three waveform design methods are verified via extensive simulations. In the considered MIMO-OFDM system, we set $N_t = N_r =4$. The transmit antennas transmit OFDM signals with 32 IFFT points, and all subcarriers are used for communication and sensing, i.e., ${{N}_{c}}=32$. The sensing channel and the communication channel are assumed to be frequency-selective MIMO channels, and the channels remain unchanged in ${{N}_{x}}=10$ symbol durations. The sensing channel and the communication channel are spatially correlated, and the correlation channel is generated using the Kronecker model \cite{b14}. The maximum correlation coefficient of the sensing channel and the communication channel are 0.5. We assumed that there are four paths and four radar targets in total.

The communication performance is characterized by spectral efficiency (bit/s/Hz), which is the communication MI (bit) on the unit frequency of unit time. The sensing performance is measured by the sensing rate (bit/s/Hz), which is the sensing MI (bit) on the unit frequency of unit time.

\begin{figure}[!t]
	\centering
	\includegraphics[width=0.5\textwidth]{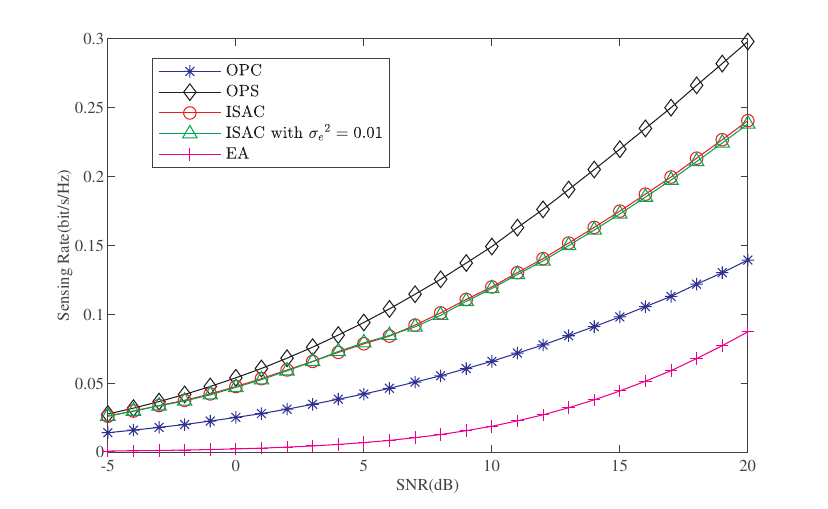}
	\caption{Sensing rate vs. SNR for different schemes.}
	\label{fig-3}
\end{figure}

\begin{table}[!t]
	\renewcommand\arraystretch{1.2}
	\caption{\label{sim_para}Simulation Parameters}
	
	\begin{center}
		\scalebox{0.9}{
		\begin{tabular}{l l l}
			\hline
			\hline
			
			{Symbol} & {Description} &{Value} \\
			
			\hline
			
			${N}_{c}$ & Number of subcarriers  & 32\\
			${N}_{t}$ & Number of transmit antenna elements & 4\\
			${{N}_{r}}$ & Number of receive antenna elements  & 4\\
			${{N}_{x}}$ & Number of OFDM symbols & 10\\
			${L}_{c}$ & Number of paths & 4\\
			${L}_{r}$ & Number of targets &4\\
			$B$ & Bandwidth &100 MHz\\
			$f_c$ & Carrier frequency &3 GHz\\
			${\sigma}_{n}^2$ &Noise variance &1\\
			${\sigma}_{h}^2$ &Mean path loss of communication channel  &(0.25,0.25,0.25,0.25)\\
			${\sigma}_{g}^2$ &Mean path loss of sensing channel &(0.25,0.25,0.25,0.25)\\
			
			\hline
			\hline
		\end{tabular}
	}
	\end{center}
\end{table}

In this section, the comparison schemes are the following four optimization schemes.
\begin{itemize}
	\item OPC: The optimization scheme that maximizes the communication MI based on the communication channel.
	
	\item OPS: The optimization scheme that maximizes sensing MI based on sensing channels.
	
	\item ISAC: The optimization scheme that maximizes the weighted sum of sensing MI and communication MI.
	
	\item EA: The equal allocation of power.

\end{itemize}

Through the analysis in Section \uppercase\expandafter{\romannumeral4}, the results obtained by the proposed scheme are the optimal solutions. As rare optimization algorithms are studied in ISAC MIMO-OFDM system, the comparison in this paper is mainly based on the EA (RA), OPC and OPS. The simulation results are obtained under the average of 4000 Monte Carlo simulations.

Fig. \ref{fig-2} 
plots the spectral efficiency of the six schemes versus the SNR, where the sensing weighting coefficient of the ISAC scheme is 0.5. It is revealed that the spectral efficiency of the six schemes gradually increases with the growth of SNR. In comparison to other schemes, the communication performance of the OPC scheme generally performs much better. In low SNR regimes, the communication performance of the ISAC scheme is slightly worse than the OPC scheme, while the OPS scheme is between ISAC and EA. In high SNR regimes, the EA scheme gradually approaches the OPC scheme and tends
to be consistent when the SNR is infinite. The reason is that when
the SNR is infinite, their MI expressions tend to be consistent. Besides, the effect of imperfect CSI is shown with the variance of channel estimation error to be 0.01, i.e., ${\sigma_e}^2=0.01$. The communication performance of ISAC scheme decreases due to the imperfect CSI. In general, as the SNR increases, the performance difference between OPC and OPS gradually increases.

Fig. \ref{fig-3} reveals the sensing rate of the six schemes versus the SNR, where the sensing weighting coefficient of the ISAC scheme is 0.5. In general, the sensing rate of the six schemes gradually increases with the increase of SNR. The OPS scheme has considerably superior sensing performance than other schemes, while the ISAC scheme has greater sensing performance than the OPC, EA, and RA schemes. As the SNR increases, the performance difference between different schemes gradually increases. Although the ISAC scheme reduces the sensing MI, it is still better than the OPC scheme. The sensing performance of the ISAC scheme also decreases slightly due to the imperfect CSI. As the ISAC scheme comprehensively considers the channel condition of sensing and communication, the optimal solution is affected which causes the sensing rate changing.  Through the analysis of Fig. \ref{fig-2} and Fig. \ref{fig-3}, it is revealed that the communication and sensing performance of the ISAC scheme are close to the communication performance of the OPC scheme and the sensing performance of the OPS scheme, respectively.

% We further evaluate the impact of the sensing weighting coefficient on the performance of sensing and communication which is analyzed in Figs. \ref{fig-4} and \ref{fig-5}.
Furthermore, we evaluate the impact of the sensing weighting coefficient on the performance of sensing and communication, which is illustrated in Figs. \ref{fig-4} and \ref{fig-5}.

\begin{figure}[!t]
	\centering
	\includegraphics[width=0.5\textwidth]{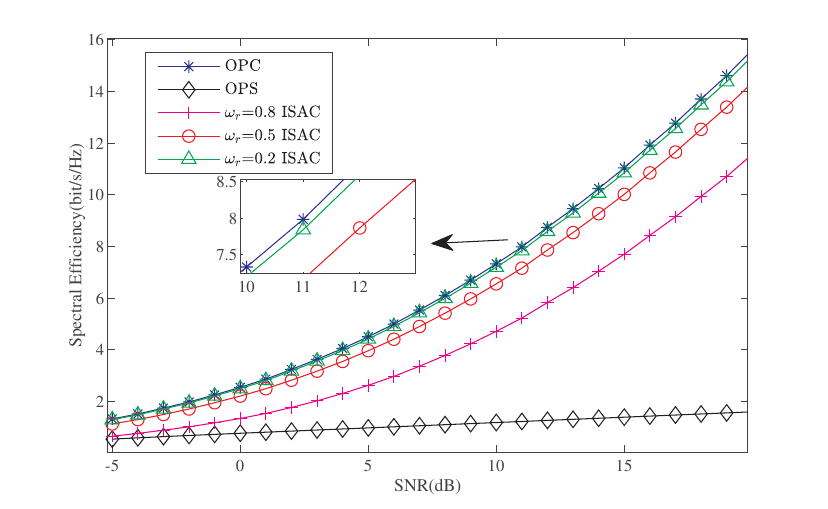}
	\caption{Spectral efficiency vs. SNR for different weighting coefficients of sensing.}
	\label{fig-4}
\end{figure}

\begin{figure}[!t]
	\centering
	\includegraphics[width=0.5\textwidth]{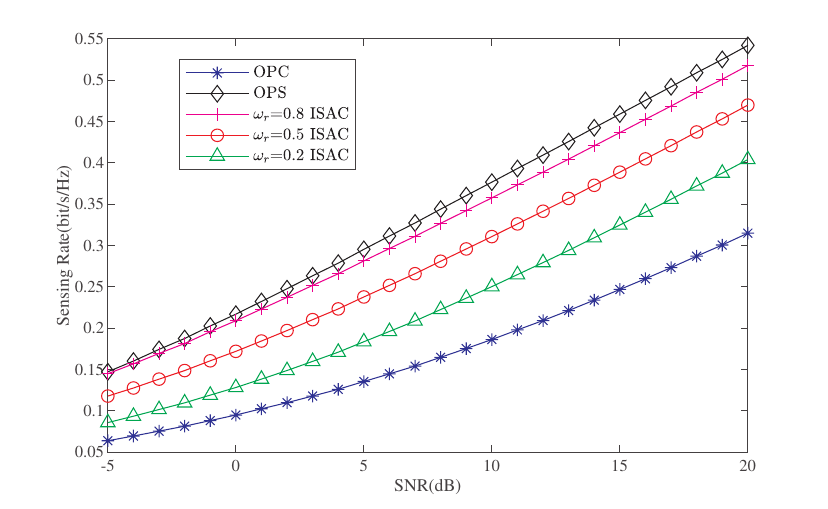}
	\caption{Sensing rate vs. SNR for different weighting coefficients of sensing.}
	\label{fig-5}
\end{figure}

Fig. \ref{fig-4} plots the spectral efficiency of the ISAC scheme with different weighting coefficients versus the SNR. When ${{\omega }_{r}}=0.2$, the OPC scheme and the ISAC scheme nearly have identical communication performance. When ${{\omega }_{r}}=0.8$, the communication performance of the ISAC scheme is almost the same as that of the OPS scheme, especially in the case of low SNR regimes. The performance difference between different schemes gradually increases with the growth of SNR. The performance when ${{\omega }_{r}}=0.5$ is in the middle. If only the communication performance is considered, the different weighting coefficients of ISAC have a great influence on the communication performance.

\begin{figure*}[!t]
	\centering
	\subfigure[\label{fig-6a}SNR=1 dB.]{\includegraphics[width=0.49\textwidth]{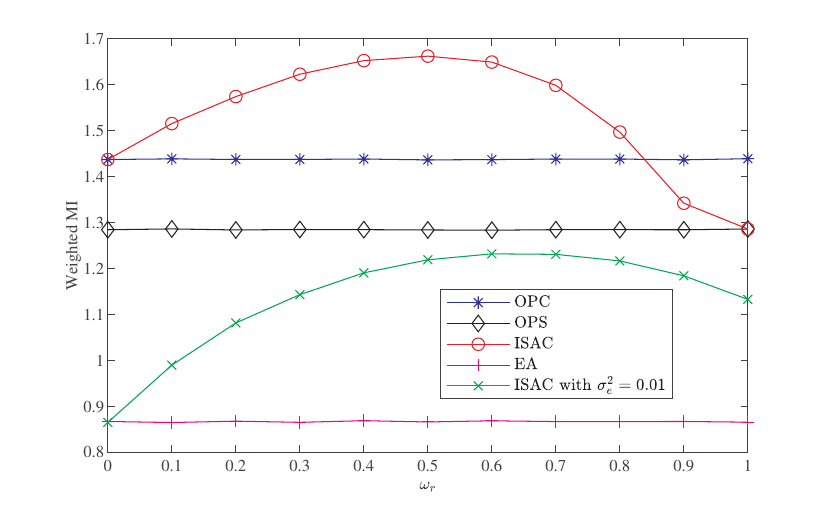}}
	%\subfigure[\label{fig-6b}SNR=10 dB.]{\includegraphics[width=0.49\textwidth]{6(b)wra.eps}}
	\subfigure[\label{fig-6b}SNR=10 dB.]{\includegraphics[width=0.49\textwidth]{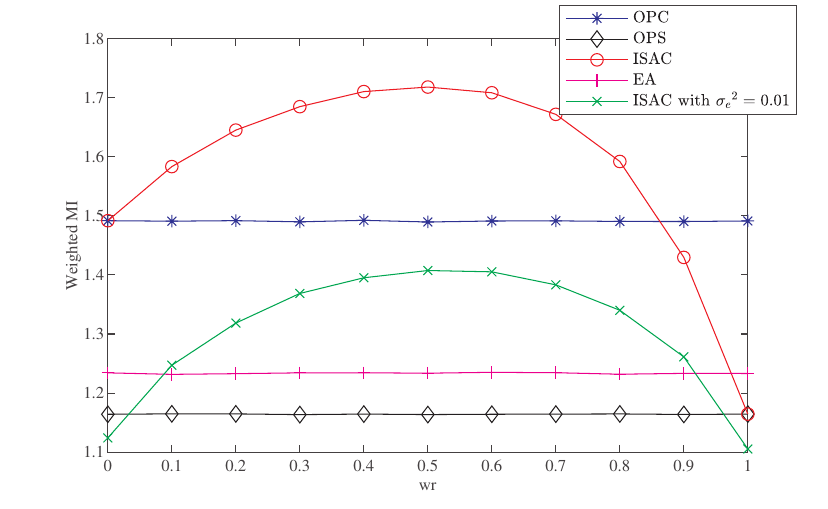}}
	\caption{ Weighted MI vs. weighting coefficient for different waveform optimization schemes }
\end{figure*}

\begin{figure}[]
	\centering
	\includegraphics[width=0.5\textwidth]{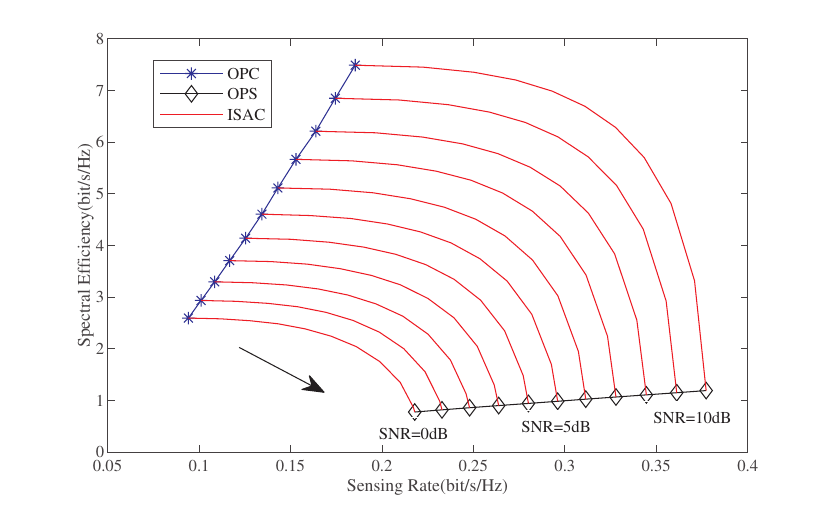}
	\caption{Trade-off curve.}
	\label{fig-7}
\end{figure}

Fig. \ref{fig-5} shows the sensing rate of ISAC with different weighting coefficients versus the SNR. When ${{\omega }_{r}}=0.8$, the sensing performance of ISAC is close to that of OPS with low SNR. As the SNR increases, the difference gets larger. When ${{\omega }_{r}}=0.2$, the sensing performance of ISAC is close to that of the OPC scheme, and the performance difference grows with the SNR. If only the sensing performance is considered, the different weighting coefficients of ISAC affect the sensing performance a lot. Figs. \ref{fig-4} and \ref{fig-5}  also show that if ${{\omega }_{r}}=0$ or ${{\omega }_{c}}=0$, the performance curve of ISAC scheme will coincide with those of OPC and OPS.

In the ISAC waveform design, the overall performance of the ISAC waveform is optimized. According to Figs. \ref{fig-2} and \ref{fig-3}, the sensing performance and the communication performance are not simultaneously optimal in the ISAC scheme. The weighted MI is used to represent the overall performance. According to (\ref{eq66}), the weighted MI of five schemes versus different weighting coefficients is simulated.

Fig. \ref{fig-6a} shows the weighted MI versus the weighting coefficient of sensing under different schemes with SNR=1 dB. According to Fig. \ref{fig-6a}, it can be found that with the weighting coefficient increasing, the weighted MI of the ISAC scheme increases at first and then decreases. When ${{\omega }_{r}}=0$ or ${{\omega }_{r}}=1$, the weighted MI of the ISAC scheme coincides with that of OPC or OPS schemes. As the weighting coefficient increases, the overall performance of the OPC, OPS, and EA schemes is almost unchanged. In general, the weighted MI of ISAC is better than that of OPC, OPS, and EA, indicating that the application of weighted MI is reasonable in the ISAC waveform design. Besides, the impact of the imperfect channel estimation is shown in Fig. \ref{fig-6a} with ${\sigma_e}^2=0.01$. It is found that the performance of the ISAC scheme decreases significantly due to the channel estimation error, since the optimization is based on the non-practical channel gains.

Fig. \ref{fig-6b} shows the same change curve as Fig. \ref{fig-6a} with SNR=10 dB. The communication performance and sensing performance of each scheme are different from those in the low SNR regimes. When the sensing weighting coefficient is close to 1, the weighted MI of the OPS
scheme decreases significantly in the high SNR regimes, and
so does the performance of the ISAC scheme, even far lower
than that of the OPC scheme. Thus, the ISAC scheme with imperfect CSI also decreases, lower than the OPS scheme when $\omega_r$ is small or large enough. The communication performance of OPS and OPC shows a huge difference in the high SNR regimes, much larger than the difference of the sensing performance between OPC and OPS. 

\begin{figure}[]
	\centering
	\includegraphics[width=0.5\textwidth]{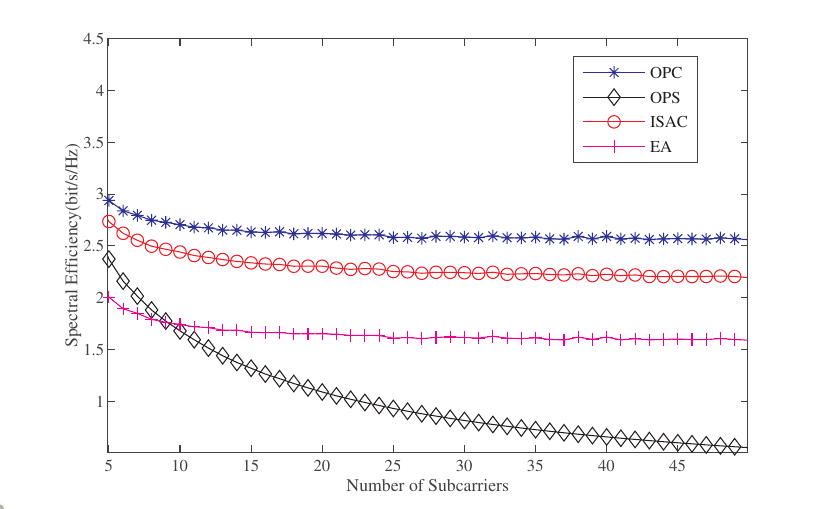}
	\caption{Spectral efficiency vs. the number of subcarriers for different waveform optimization schemes.}
	\label{fig-8}
\end{figure}

\begin{figure}[]
	\centering
	\includegraphics[width=0.5\textwidth]{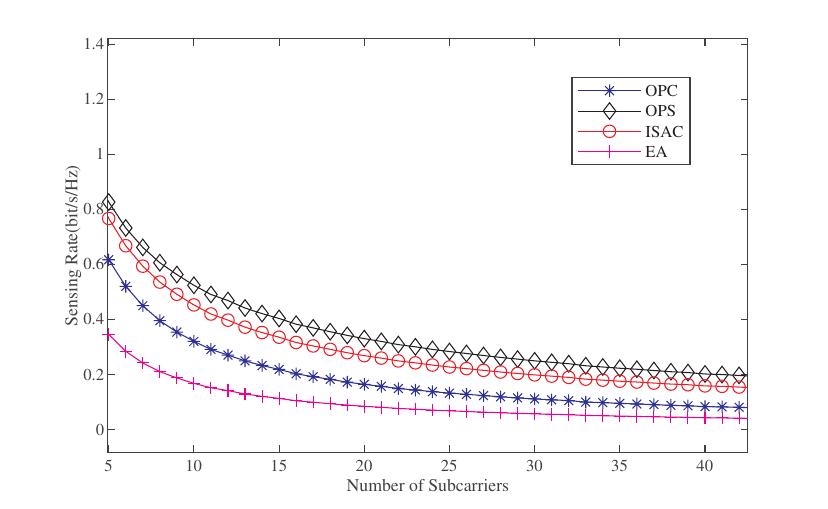}
	\caption{Sensing rate vs. the number of subcarriers for different waveform optimization schemes.}
	\label{fig-9}
\end{figure}

\begin{figure}[!t]
	\centering
	\includegraphics[width=0.5\textwidth]{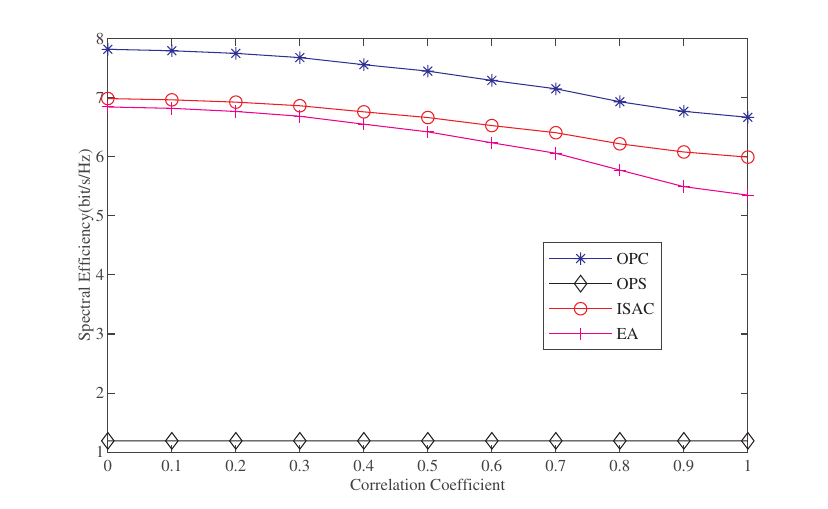}
	\caption{Spectral efficiency vs. correlation coefficient for different waveform optimization schemes.
		\label{fig-10}}
\end{figure}

\begin{figure}[!t]
	\centering
	\includegraphics[width=0.5\textwidth]{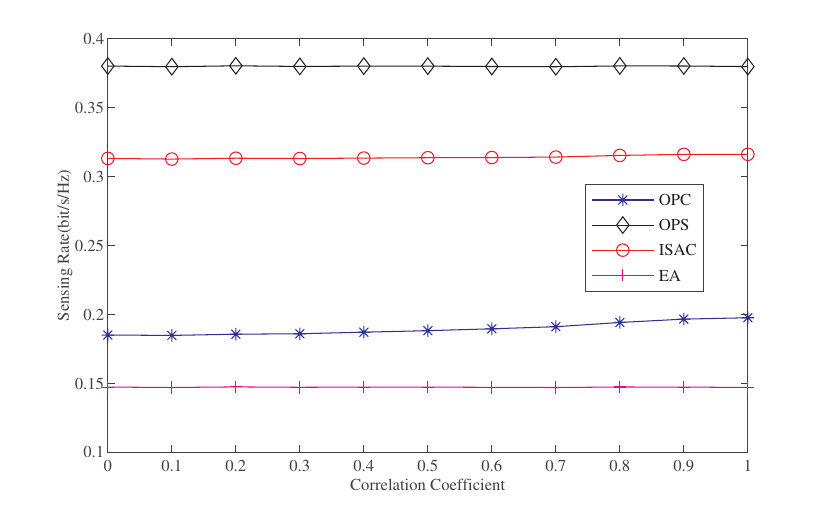}
	\caption{Sensing rate vs. correlation coefficient for different waveform optimization schemes.
		\label{fig-11}}
\end{figure}

\begin{figure}[t]
	\centering
	\includegraphics[width=0.5\textwidth]{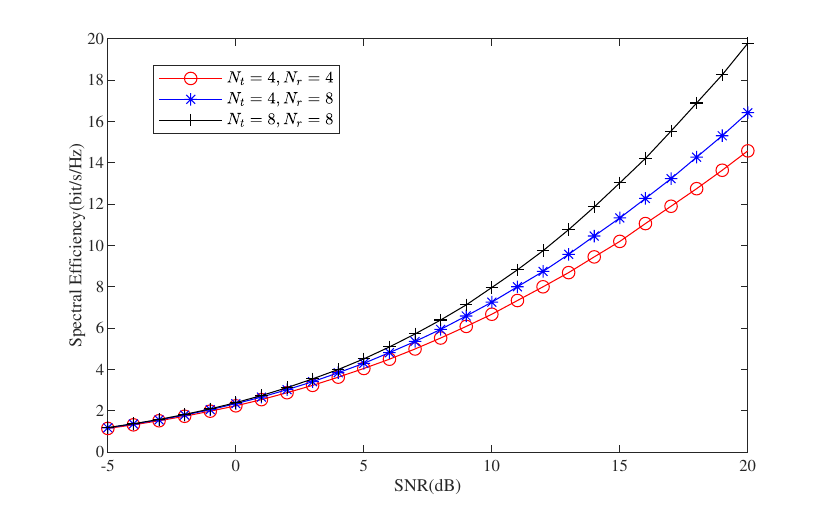}
	\caption{Spectral efficiency vs. SNR for ISAC scheme of different number of antennas.
		\label{fig-12}}
\end{figure}

\begin{figure}[t]
	\centering
	\includegraphics[width=0.5\textwidth]{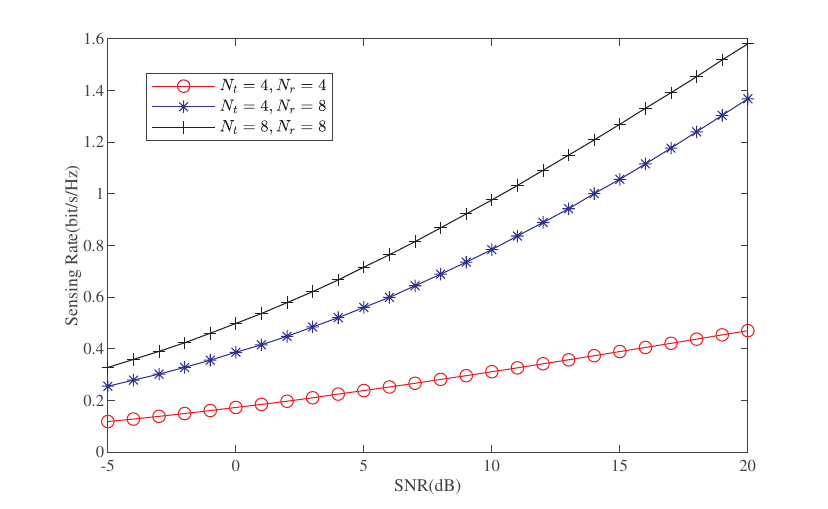}
	\caption{Sensing rate vs. SNR for ISAC scheme of different number of antennas.
		\label{fig-13}}
\end{figure}

Fig. \ref{fig-7} shows the trade-off between the sensing and communication for the OPC, OPS and ISAC schemes under
different SNRs. The ISAC scheme has a sensing weighting coefficient from 0 to 1 in the direction of the arrow. With the increase of SNR, the sensing performance and communication performance are improved, and the difference in the weighting coefficient causes the ISAC scheme to change between the OPC scheme and the OPS scheme. It is noted that there is no comparison among the weighted sum with different weighting factors. An appropriate weighting coefficient can be selected through Fig. \ref{fig-7} to meet the certain performance requirement of communication and sensing. Fig. \ref{fig-7} also reveals that it is impossible to attain optimal communication and sensing performance simultaneously, only possible to make a trade-off between communication and sensing.

With OFDM signal, the bandwidth $B$ (Hz) is divided into ${{N}_{c}}$ subcarriers, which satisfies $B={{N}_{c}}\Delta f$, where  $\Delta f$ is the subcarrier spacing. In Sec. \uppercase\expandafter{\romannumeral2}, we prove that the subcarrier spacing is related to the subcarrier correlation, so that the number of subcarriers is used to represent the subcarrier correlation coefficient.

Fig. \ref{fig-8} shows the spectral efficiency of different schemes versus the number of subcarriers, where SNR=1 dB and the number of subcarriers ranges from 5 to 50. It is revealed from Fig. 8 that with the increasing number of subcarriers, the spectral efficiency of OPC, ISAC, and EA schemes is basically unchanged, while the spectral efficiency of the OPS scheme is decreasing. This is due to the fact that the communication channel correlation of different subcarriers has been omitted in Sec. \uppercase\expandafter{\romannumeral2}.

Fig. \ref{fig-9} shows the sensing rate of different schemes versus the number of subcarriers, where the parameters are identical to those in Fig. 8. As the number of subcarriers grows, the sensing rates of OPS and ISAC schemes decrease, which is related to the correlation of subcarriers in the sensing channel. Meanwhile, the sensing rates of the OPC and EA schemes first decrease with an increase in the number of subcarriers, and then gradually remain unchanged when the number of subcarriers is large.

Finally, the spectral efficiency and sensing rate of different schemes versus different spatial correlation coefficients are simulated. The correlation coefficient increases with the decrease of the angle spread and the increase of antenna spacing~\cite{b15}. The simulation is conducted using the maximum correlation coefficient of the changing communication channel.

Fig. \ref{fig-10} illustrates the spectral efficiency of different schemes versus the spatial correlation coefficient. Except for the OPS scheme which remains basically unchanged with the increase of the correlation coefficient, the spectral efficiency of the other schemes decreases slightly with the increasing correlation coefficient. This is due to the fact that, in the high
SNR regimes, the OPC scheme and the EA scheme tend to be
consistent.

Fig. \ref{fig-11} shows the sensing rate of different schemes versus the spatial correlation coefficient. It can be found that different from Fig. \ref{fig-10}, with the increase of the spatial correlation coefficient, the sensing rates of all the schemes fluctuate slightly. This is due to the fact that only the maximum correlation coefficient of the communication channel changes in the simulation, merely changing the OPC scheme slightly. 
In addition, the maximum correlation coefficient of the communication channel has a limited effect on the sensing performance
of the ISAC scheme.

Fig. \ref{fig-12} shows the spectral efficiency of the ISAC scheme versus the SNR, where ${{\omega }_{r}}=0.5$. The number of antennas refers to the number of sub-channels (${{N}_{t}}\times {{N}_{r}}$). It is found that the communication performance improves with the increase of the number of antennas, and the performance difference enlarges with the increase of SNR.

Fig. \ref{fig-13} shows the sensing rate of the ISAC scheme versus the SNR, where ${{\omega }_{r}}=0.5$. It is found that the results are similar to those in Fig. \ref{fig-12}. The difference is that the sensing performance improves significantly with the increase of the number of receive antennas, as the ISAC scheme proposed in this paper is based on the spatial uncorrelation of each receive antenna.

\section{Conclusion}
In this paper, the sensing and communication MI for MIMO-OFDM ISAC system 
with subcarrier correlation and spatial correlation were studied.
More specifically, the ISAC waveform optimization schemes were investigated for 
maximizing sensing MI, communication MI, and their weighted sum, respectively. 
The derived MI optimization solutions are validated by the simulation results.
We disclose the trade-off in performance between communication and sensing, showing how balanced performance can be achieved with carefully designed MIMO-OFDM ISAC waveform. 
Our work can be further extended by considering more practical scenarios with imperfect CSI and the resource allocation between training and data payload signals.

\vspace{0.5 mm}
\begin{IEEEbiography}[{\includegraphics[width=1.1in,height=1.25in,clip,keepaspectratio]{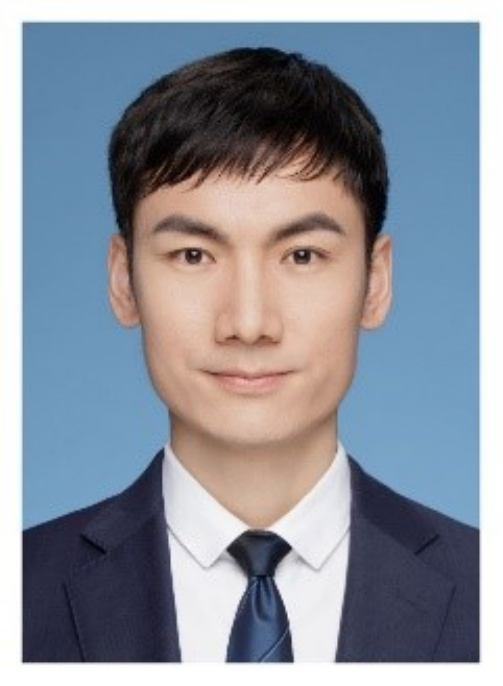}}]{Zhiqing Wei} (Member, IEEE) received the B.E. and Ph.D. degrees from the Beijing University of Posts and Telecommunications (BUPT), Beijing, China, in 2010 and 2015, respectively. He is an Associate Professor with BUPT. He has authored one book, three book chapters, and more than 50 papers. His research interest is the performance analysis and optimization of intelligent machine networks. He was granted the Exemplary Reviewer of IEEE WIRELESS COMMUNICATIONS LETTERS in 2017, the Best Paper Award of WCSP 2018. He was the Registration Co-Chair of IEEE/CIC ICCC 2018, the publication Co-Chair of IEEE/CIC ICCC 2019 and IEEE/CIC ICCC 2020.
\end{IEEEbiography}

\begin{IEEEbiography}[{\includegraphics[width=1.1in,height=1.25in,clip,keepaspectratio]{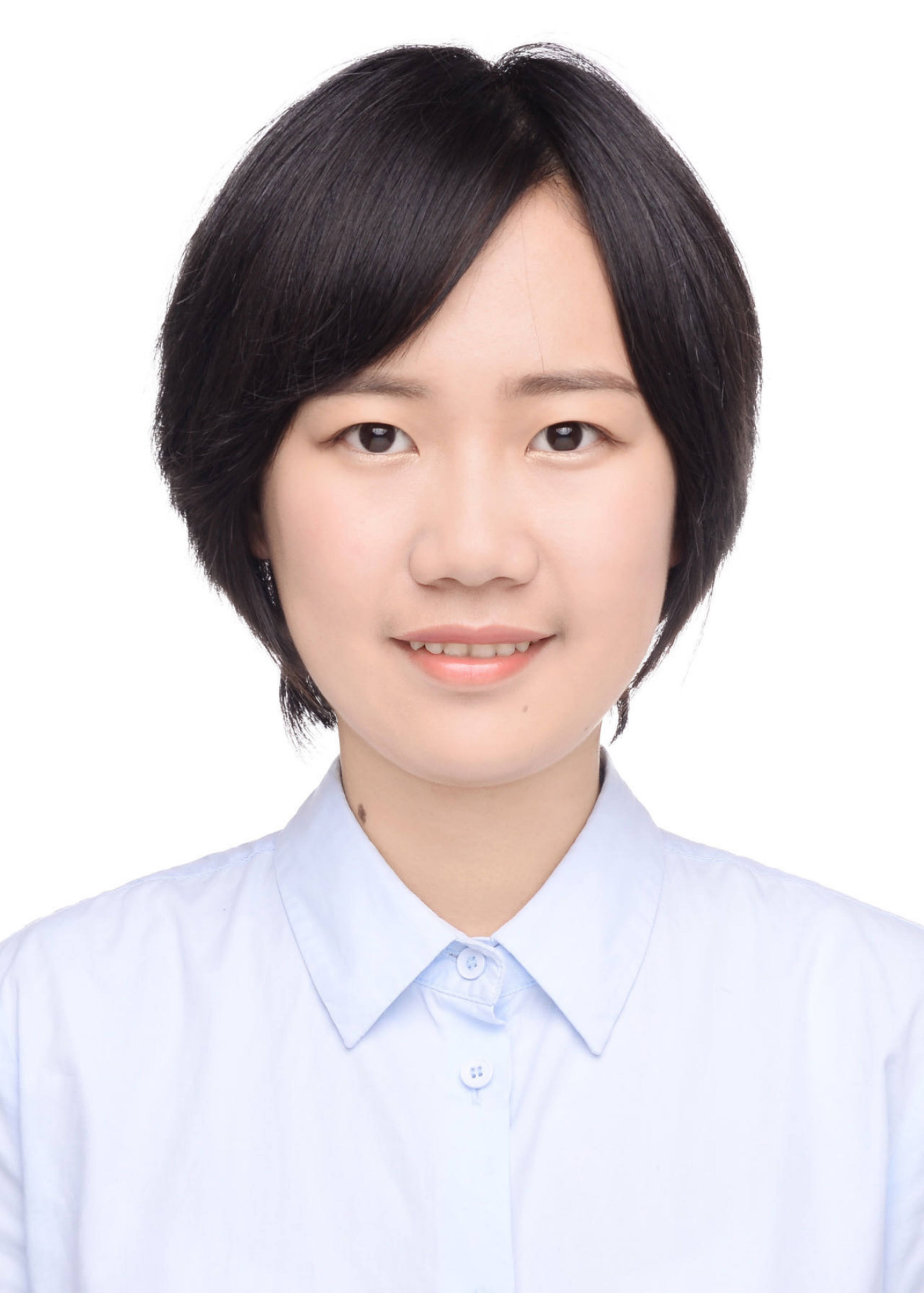}}]{Jinghui Piao} (Student Member, IEEE) received the B.E. degree from the Beijing Jiaotong University (BJTU), Beijing, China, in 2022. She is currently pursuing the master’s degree in Beijing University of Posts and Telecommunications (BUPT), Beijing, China. Her research interests are in integrated sensing and communication.
\end{IEEEbiography}

\begin{IEEEbiography}[{\includegraphics[width=1in,height=2in,clip,keepaspectratio]{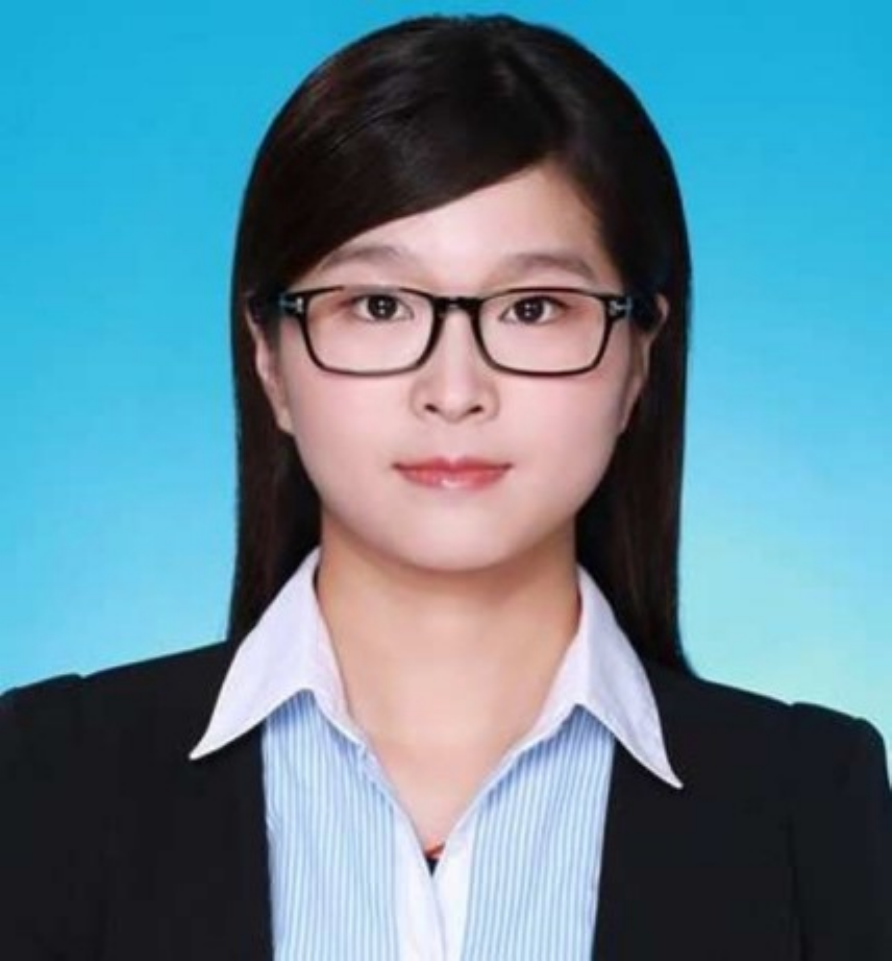}}]{Xin Yuan} (S'16-M'19) received Ph.D. degrees from Beijing University of Posts and Telecommunications (BUPT), Beijing, China, and the University of Technology Sydney (UTS), Sydney, Australia, in 2019 and 2020, respectively. Her research interests include machine learning and optimization, and their applications to UAV networks and intelligent systems.
\end{IEEEbiography}

\begin{IEEEbiography}[{\includegraphics[width=1.25in,height=1.4in,clip,keepaspectratio]{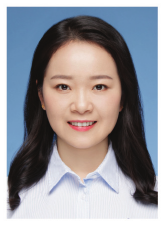}}]{Huici Wu} (Member, IEEE) received the Ph.D degree from Beijing University of Posts and Telecommunications (BUPT), Beijing, China, in 2018. From 2016 to 2017, she visited the Broadband Communications Research (BBCR) Group, University of Waterloo, Waterloo, ON, Canada. She is now an Associate Professor at BUPT. Her research interests are in the area of wireless communications and networks, with current emphasis on collaborative air-to-ground communication and wireless access security.
\end{IEEEbiography}

\begin{IEEEbiography}[{\includegraphics[width=1.1in,height=1.25in,clip,keepaspectratio]{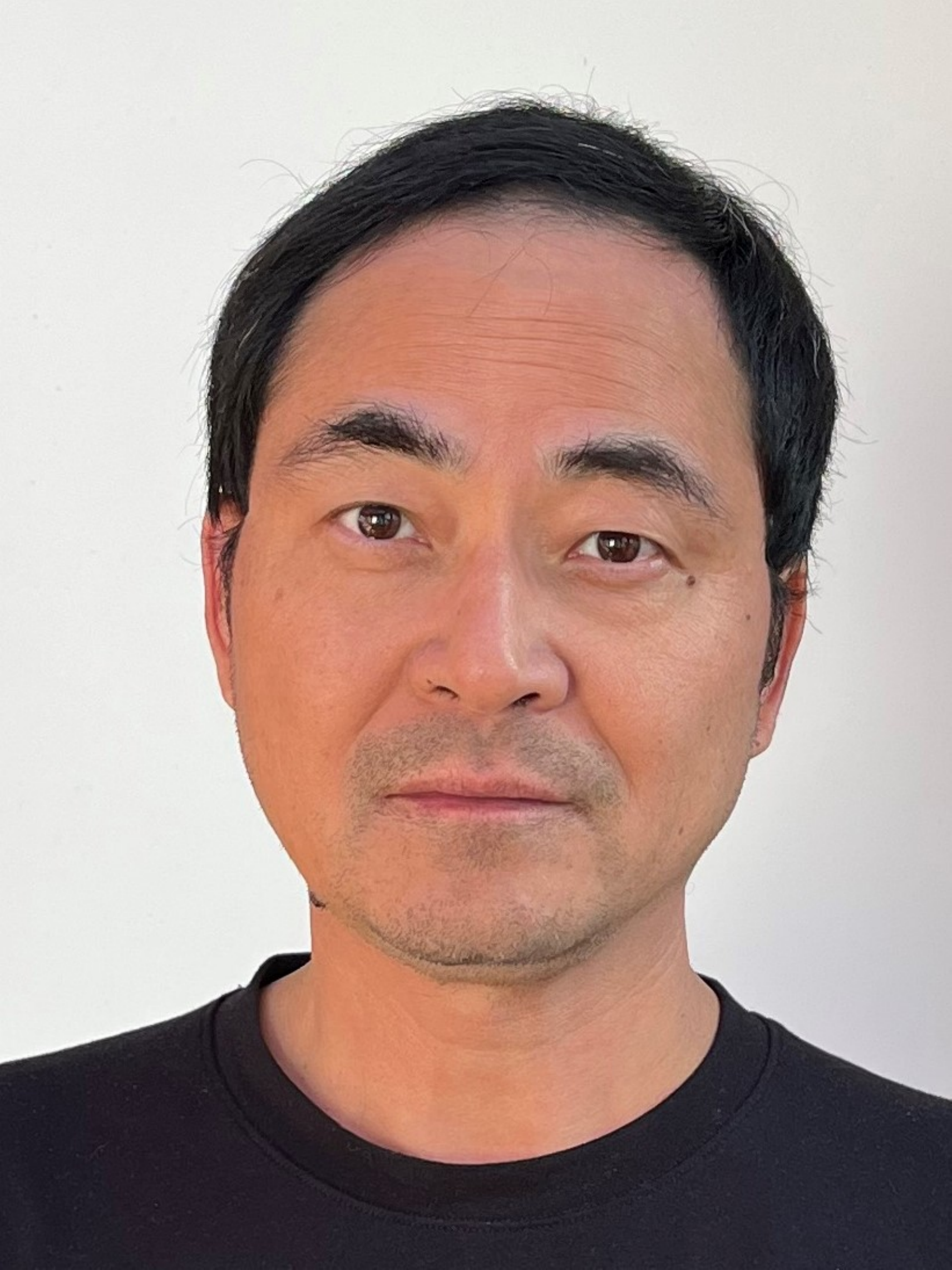}}]{J. Andrew Zhang} (M'04-SM'11) received the B.Sc. degree from Xi'an JiaoTong University, China, in 1996, the M.Sc. degree from Nanjing University of Posts and Telecommunications, China, in 1999, and the Ph.D. degree from the Australian National University, Australia, in 2004. Currently, he is a Professor in the School of Electrical and Data Engineering, University of Technology Sydney, Australia.
	
	Dr. Zhang's research interests are in the area of signal processing for wireless communications and sensing. He has published more than 270 papers in leading Journals and conference proceedings, and has won 5 best paper awards for his work including in IEEE ICC2013. He is a recipient of CSIRO Chairman's Medal and the Australian Engineering Innovation Award in 2012 for exceptional research achievements in multi-gigabit wireless communications.
\end{IEEEbiography}

\begin{IEEEbiography}[{\includegraphics[width=1.1in,height=1.25in,clip,keepaspectratio]{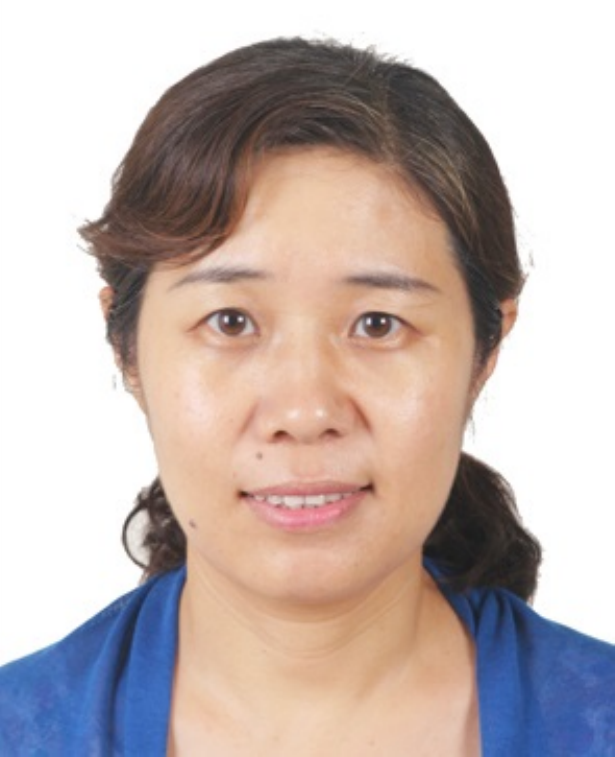}}]{Zhiyong Feng} (M'08-SM'15) received her B.E., M.E., and Ph.D. degrees from Beijing University of Posts and Telecommunications (BUPT), Beijing, China. She is a professor at BUPT, and the director of the Key Laboratory of the Universal Wireless Communications, Ministry of Education, P.R.China. She is a senior member of IEEE, vice chair of the Information and Communication Test Committee of the Chinese Institute of Communications (CIC). Currently, she is serving as Associate Editors-in-Chief for China Communications, and she is a technological advisor for international forum on NGMN. Her main research interests include wireless network architecture design and radio resource management in 5th generation mobile networks (5G), spectrum sensing and dynamic spectrum management in cognitive wireless networks, and universal signal detection and identification.
\end{IEEEbiography}

\begin{IEEEbiography}[{\includegraphics[width=1in,height=1.15in,clip,keepaspectratio]{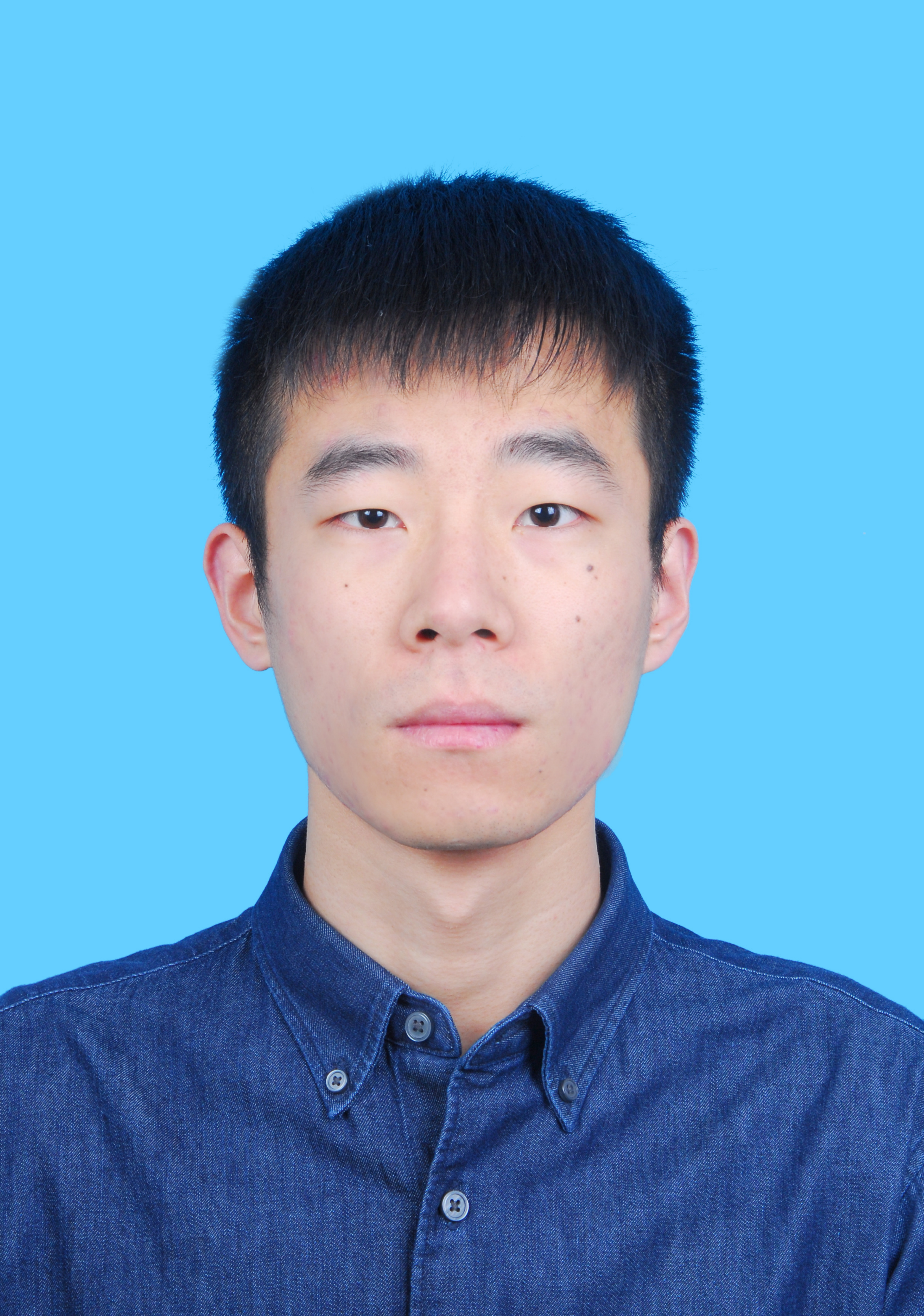}}]{Lin Wang} (Student Member, IEEE) received the B.E. degree from the Beijing University of Posts and Telecommunications (BUPT), Beijing, China, in 2021. He is currently pursuing the Ph.D. degree in BUPT. His research interests include stochastic geometry and integrated sensing and communication.
\end{IEEEbiography}

\begin{IEEEbiography}[{\includegraphics[width=1in,height=1.15in,clip,keepaspectratio]{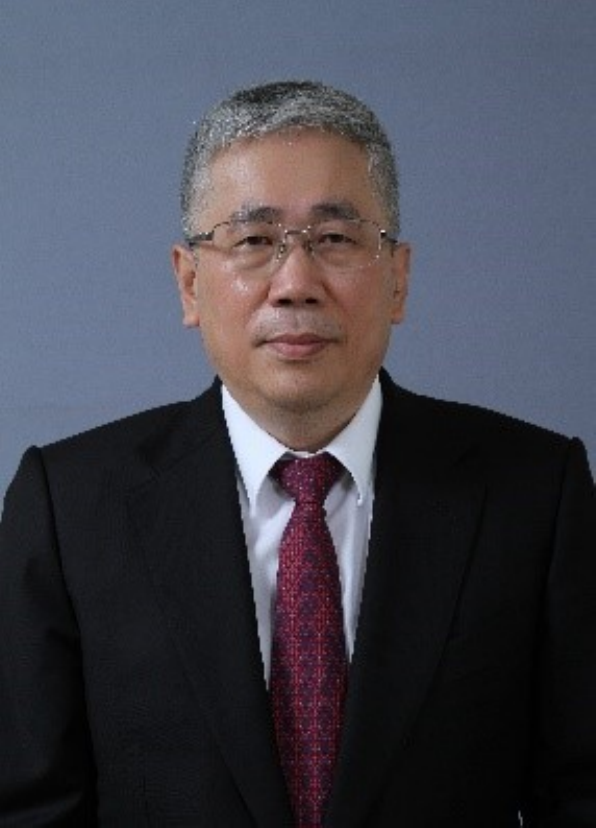}}]{Ping Zhang} (M'07-SM'15-F'18) received the Ph.D. degree from the Beijing University of Posts and Telecommunications (BUPT), Beijing, China, in 1990. He is currently a Professor of School of Information and Communication Engineering at BUPT , the Director of the State Key Laboratory of Networking and Switching Technology, the Director of the Department of Broadband Communication of Peng Cheng Laboratory, a member of IMT-2020 (5G) Experts Panel, and a member of Experts Panel for China’s 6G development. He served as a Chief Scientist of National Basic Research Program (973 Program), an expert in Information Technology Division of National High-Tech R\&D Program (863 Program), and a member of Consultant Committee on International Cooperation of National Natural Science Foundation of China. Prof. Zhang is an Academician of the Chinese Academy of Engineering. His research interests mainly focus on wireless communication.
\end{IEEEbiography}


\begin{thebibliography}{00}

\bibitem{wang2021}
Z. Wang, K. Han, J. Jiang, Z. Wei, G. Zhu, Z. Feng, J.
Lu, and C. Meng.  ``Symbiotic sensing and communications
towards 6G: Vision, applications, and technology trends,''
in \emph{Vehicular Technology
	Conference (IEEE VTC 2021 fall)}, Sep. 2021.

\bibitem{liu2022}
F. Liu \emph{et al.}, ``Integrated Sensing and Communications: Toward Dual-Functional Wireless Networks for 6G and Beyond,'' in \emph{IEEE J. Sel. Areas Commun.}, vol. 40, no. 6, pp. 1728-1767, Jun. 2022.

\bibitem{feng2020}
Z. Feng, Z. Fang, Z. Wei, X. Chen, Z. Quan, and D.
Ji.  ``Joint radar and communication: A survey,'' \emph{China Commun.}, vol. 17, no. 1, pp. 1–27, Jan. 2020.

\bibitem{imt2030}
Introduction to 6G, 2022. [Online.] Available: https://www.tonex.com/training-courses/introduction-to-6g-imt-2030/

\bibitem{zhang2021}
J. A. Zhang, F. Liu, C. Masouros, R. W. Heath and A.
Petropulu, ``An overview of signal processing techniques
for joint communication and radar sensing,'' \emph{IEEE J. Sel. Top. Signal Process.}, pp. 1–1, Sep.
2021.

\bibitem{pin2021}
D. K. Pin Tan \emph{et al.}, ``Integrated Sensing and Communication in 6G: Motivations, Use Cases, Requirements, Challenges and Future Directions,'' \emph{2021 1st IEEE International Online Symposium on Joint Communications \& Sensing (JC\&S)}, 2021, pp. 1-6.
\bibitem{kay1993}
S. M. Kay, \emph{Fundamentals of Statistical Signal Processing}, vol. 10. Englewood Cliffs, NJ, USA: PTR Prentice-Hall, 1993, Art. no. 151045.
\bibitem{b1} 
A. R. Chiriyath, B. Paul, G. M. Jacyna and D. W. Bliss, ``Inner Bounds on Performance of Radar and Communications Co-Existence,'' \emph{IEEE Trans. Signal Process.}, vol. 64, no. 2, pp. 464-474, Jan. 2016.

\bibitem{b2} 
P. Kumari, D. H. N. Nguyen and R. W. Heath, ``Performance trade-off in an adaptive IEEE 802.11AD waveform design for a joint automotive radar and communication system,'' \emph{2017 IEEE International Conference on Acoustics, Speech and Signal Processing (ICASSP)}, pp. 4281-4285, 2017.

\bibitem{b3} 
W. Zhang, S. Vedantam, and U. Mitra, ``Joint transmission and state estimation: A constrained channel coding approach,'' \emph{ IEEE Trans. Inf. Theory}, vol. 57, no. 10, pp. 7084–7095, Oct. 2011.

\bibitem{b4} 
M. Kobayashi, G. Caire and G. Kramer, ``Joint State Sensing and Communication: Optimal Tradeoff for a Memoryless Case,'' \emph{2018 IEEE International Symposium on Information Theory (ISIT)}, 2018, pp. 111-115. 

\bibitem{b5} 
C. Ouyang, Y. Liu and H. Yang, ``Integrated Sensing and Communications: A Mutual Information-Based Framework,'' arXiv:2208.04260, 2022.

\bibitem{b6} 
M. R. Bell, ``Information theory and radar waveform design,'' \emph{IEEE Trans. Inf. Theory}, vol. 39, no. 5, pp. 1578-1597, Sept. 1993.

\bibitem{b7} 
Y. Yang and R. S. Blum, ``MIMO radar waveform design based on mutual information and minimum mean-square error estimation,'' \emph{Abbreviation	Title
	IEEE Trans. Aerosp. Electron. Syst.}, vol. 43, no. 1, pp. 330-343, January 2007.

\bibitem{b8} 
T. Tian, T. Zhang, L. Kong, G. Cui and Y. Wang, ``Mutual Information based Partial Band Coexistence for Joint Radar and Communication System,'' \emph{2019 IEEE Radar Conference (RadarConf)}, 2019, pp. 1-5.

\bibitem{b12} 
X. Yuan \emph{et al.}, ``Spatio-Temporal Power Optimization for MIMO Joint Communication and Radio Sensing Systems With Training Overhead," \emph{IEEE Trans. Veh. Technol.}, vol. 70, no. 1, pp. 514-528, Jan. 2021.

\bibitem{liutransmission} 
J. Liu and M. Saquib, ``Transmission Design for a Joint MIMO Radar and MU-MIMO Downlink Communication System,'' \emph{2018 IEEE Global Conference on Signal and Information Processing (GlobalSIP)}, Anaheim, CA, USA, 2018, pp. 196-200.

\bibitem{b9} 
S. Zhu, X. Li, R. Yang and X. Zhu, ``A low probability of intercept OFDM radar communication waveform design method,'' \emph{2021 IEEE International Conference on Consumer Electronics and Computer Engineering (ICCECE)}, 2021, pp. 653-657.

\bibitem{b10} 
A. Ahmed, Y. D. Zhang, A. Hassanien and B. Himed, ``OFDM-based Joint Radar-Communication System: Optimal Sub-carrier Allocation and Power Distribution by Exploiting Mutual Information,'' \emph{2019 53rd Asilomar Conference on Signals, Systems, and Computers(ACSSC)}, 2019, pp. 559-563.

\bibitem{b11} 
Y. Liu, G. Liao, J. Xu, Z. Yang and Y. Zhang, ``Adaptive OFDM Integrated Radar and Communications Waveform Design Based on Information Theory,'' \emph{IEEE Commun. Lett.}, vol. 21, no. 10, pp. 2174-2177, Oct. 2017.


\bibitem{ouyang2022}
C. Ouyang, Y. Liu and H. Yang, ``Performance of Downlink and Uplink Integrated Sensing and Communications (ISAC) Systems,'' \emph{IEEE Wireless Communications Letters}, vol. 11, no. 9, pp. 1850-1854, Sept. 2022.

\bibitem{ts2022}
T. S. Rappaport, \emph{Wireless Communications:
	Principles and Practice}, 2nd ed. Upper
Saddle River, NJ, USA: Prentice-Hall, 2002.

\bibitem{Jafar}
S. A. Jafar and A. Goldsmith, "Multiple-antenna capacity in correlated Rayleigh fading with channel covariance information," in IEEE Transactions on Wireless Communications, vol. 4, no. 3, pp. 990-997, May 2005.

\bibitem{b15}  H. B{\"o}lcskei, ``Principles of MIMO-OFDM wireless systems," in M.Ibnkahla (Ed.), \textit{Signal Processing for Mobile Communications Handbook.} CRC Press, 2004.

\bibitem{cao}
W. Cao, X. Li, W. Hu, J. Lei and W. Zhang, ``OFDM reference signal reconstruction exploiting subcarrier-grouping-based multi-level Lloyd-Max algorithm in passive radar systems,'' \emph{IET Radar Sonar Navig.}, vol. 11, no. 5, pp. 873-879, 2017.

\bibitem{zhu}
Xiaolong Zhu and Jinyin Xue, ``On the Correlation of Subcarriers in Grouped Linear Constellation Precoding OFDM Systems Over Frequency Selective Fading," 2006 IEEE 63rd Vehicular Technology Conference, Melbourne, VIC, Australia, 2006, pp. 1431-1435.

\bibitem{zhou}
W. Zhou, R. Zhang, G. Chen and W. Wu, ``Integrated Sensing and Communication Waveform Design: A Survey," in IEEE Open Journal of the Communications Society, vol. 3, pp. 1930-1949, 2022.

\bibitem{caox}
X. Cao, L. Tang, F. Shen, Y. Zhang, F. Yan and C. Wang, ``Robust OFDM Shared Waveform Design and Resource Allocation for the Integrated Sensing and Communication System," 2023 IEEE Wireless Communications and Networking Conference (WCNC), Glasgow, United Kingdom, 2023, pp. 1-6.


\bibitem{dong}
F. Dong et al., ``Waveform Design for Communication-Assisted Sensing in 6G Perceptive Networks,'' arXiv preprint arXiv:2305.11399, 2023.


\bibitem{b13} 
G. T. Gilber, ``Positive definite matrices and Sylvester’s criterion," \emph{Am. Math. Mon.}, vol. 98, pp. 44–46, 1991.

\bibitem{boyd}
S. P. Boyd and L. Vandenberghe, Convex optimization. Cambridge university press, 2004.

\bibitem{b14} 
C. Oestges, ``Validity of the Kronecker Model for MIMO Correlated Channels," \emph{2006 IEEE 63rd Vehicular Technology Conference(VTC-2006)}, 2006, pp. 2818-2822.


\end{thebibliography}
\end{document}